\documentclass[preprint]{imsart}

\usepackage{xr}
\makeatletter

\newcommand*{\addFileDependency}[1]{
\typeout{(#1)}
\@addtofilelist{#1}
\IfFileExists{#1}{}{\typeout{No file #1.}}
}\makeatother

\newcommand*{\myexternaldocument}[1]{%
\externaldocument{#1}%
\addFileDependency{#1.tex}%
\addFileDependency{#1.aux}%
}

\myexternaldocument{SUPP-v2}

\RequirePackage{amsthm,amsmath,amsfonts,amssymb}
\RequirePackage[authoryear]{natbib}
\RequirePackage[colorlinks,citecolor=blue,urlcolor=blue]{hyperref}
\RequirePackage{graphicx}

\RequirePackage[boxed]{algorithm2e}
\RequirePackage{float}
\RequirePackage{soul}

\usepackage{booktabs}
\usepackage[table,xcdraw]{xcolor}

\usepackage{sidecap}
\sidecaptionvpos{figure}{c}

\usepackage{filecontents}
\usepackage{multibib}
\newcites{supp}{Supplementary References}

\startlocaldefs

\theoremstyle{plain}

\theoremstyle{remark}

\newtheorem{proposition}{Proposition}

\endlocaldefs

\begin{document}

\begin{frontmatter}
\title{Communication network dynamics in a large organizational hierarchy}
\runtitle{Communication network dynamics}

\begin{aug}
\author[A]{\fnms{Nathaniel} \snm{Josephs}},
\author[B]{\fnms{Sida} \snm{Peng}},
\and
\author[C,D,E,F,G]{\fnms{Forrest W.} \snm{Crawford}}

\runauthor{Josephs et al.}

\address[A]{Department of Statistics, North Carolina State University}
\address[B]{Office of Chief Economist, Microsoft Research}
\address[C]{Department of Biostatistics, Yale School of Public Health}
\address[D]{Department of Statistics \& Data Science, Yale University}
\address[E]{Department of Ecology \& Evolutionary Biology, Yale University}
\address[F]{Yale School of Management}
\address[G]{RAND Corporation}
\end{aug}

\begin{abstract}
Most businesses impose a supervisory hierarchy on employees to facilitate management, decision-making, and collaboration, yet routine inter-employee communication patterns within workplaces tend to emerge more naturally as a consequence of both supervisory relationships and the needs of the organization. What then is the relationship between a formal organizational structure and the emergent communications between its employees? Understanding the nature of this relationship is critical for the successful management of an organization. While scholars of organizational management have proposed theories relating organizational trees to communication dynamics, and separately, network scientists have studied the topological structure of communication patterns in different types of organizations, existing empirical analyses are both lacking in representativeness and limited in size. In fact, much of the methodology used to study the relationship between organizational hierarchy and communication patterns (and much of what is known about this relationship) comes from analyses of the Enron email corpus, reflecting a uniquely dysfunctional corporate environment. In this paper, we develop new methodology for assessing the relationship between organizational hierarchy and communication dynamics and apply it to Microsoft Corporation, currently the highest valued company in the world, consisting of approximately 200,000 employees divided into 88 teams, organizational trees rooted at the senior leadership level. This reveals distinct communication network structures within and between teams. We then characterize the relationship of routine employee communication patterns to these team supervisory hierarchies, while empirically evaluating several theories of organizational management and performance. To do so, we propose new measures of communication reciprocity and new shortest-path distances for trees to track the frequency of messages passed up, down, and across the organizational hierarchy. By describing how communication clusters around the formal organization, we reveal the emergent communication dynamics between employees and the crucial role of position in the hierarchy.
\end{abstract}

\begin{keyword}
\kwd{Communication Dynamics}
\kwd{Email Network}
\kwd{Organizational Hierarchy}
\kwd{Reciprocity}
\kwd{Reporting Distance}
\kwd{Path Analysis}
\kwd{Latent Tree}
\end{keyword}

\end{frontmatter}


\section{Introduction}

Most organizations impose a supervisory hierarchy on their employees to facilitate management, decision-making, and collaboration \citep{robbins2004organizational,galbraith2008organization}.
In contrast, routine interpersonal communication patterns within workplaces tend to emerge more naturally, as a consequence of both supervisory relationships and the needs of the organization.
``If the formal organization is the skeleton of a company," \citet{krackhardt1993informal} analogize, ``the informal is the central nervous system driving the collective thought processes, actions, and relations of its business units.''
Understanding the relationship between these two networks is crucial to successful management of an organization \citep{cross2007role}.

Scholars of organizational behavior and management have proposed theories relating social network features to measures of business performance \citep{tichy1979social,krackhardt1988informal} and researchers have described strategies to measure and improve communication networks within an organization \citep{cross2007role,nielsen2016work}.
Some of this work makes explicit claims about how communication network features relate to organizational effectiveness.
For example, \citet{krackhardt1988informal} conjecture that effective organizations have a higher rate of between-team than within-team communication.
Alternatively, \citet{hatch1997relations} and \citet{holtzhausen2002effects} argue that communication is better in decentralized organizations in which ``authority is dispersed downward in the hierarchy."
Several authors claim that the highest-performing employees have more communication network connections \citep{cross2002making, nielsen2016work}.

To evaluate these theories, management researchers, computer scientists, and network scientists have proposed frameworks and methodologies for studying latent, or unmeasured, hierarchical structure within interaction or communication networks \citep{capobianco1979strength,palus2011evaluation,fire2016organization}.
One objective is to use communication network features to classify individuals by rank or job title in the corresponding organizational hierarchy, thereby uncovering network correlates of institutional position. 
A common approach is computing network measures of centrality for employees, then ranking them in order of importance.
For example, \citet{shetty2004enron} use entropy to measure importance of employees, and \citet{namata2006inferring} use variations in email counts to rank employees.
\citet{rowe2007automated} rank individuals based on a combination of communication flow statistics such as average response time and topological statistics such as the number of communication network cliques to which each employee belongs.
\citet{hossain2009effect} and \citet{michalski2011matching} use classical centrality measures to rank employees, and \citet{wang2013analyzing} adapt the PageRank algorithm \citep{page1999pagerank} for email networks, which they call HumanRank.
\citet{zhang2009analyzing} assess supervised clustering techniques for predicting missing job titles based on the number of emails sent and received by each employee.
\citet{dong2015inferring} find that ``structural holes'' in the communication network are useful for identifying ``high-status" individuals, and that these individuals form a ``rich club" community of employees that is connected, balanced, and dense.
\citet{nurek2020combining} show that auxiliary employee information, such as number of days of overtime worked, can improve classification of job title.

Beyond classifying employees by rank or job title, there are three methods in the literature that attempt to infer a tree-like structure (ideally the entire organizational hierarchy) using inter-employee communication data. \citet{creamer2007segmentation} segment employees into hierarchical levels, then link them by an edge when two employees exchanged at least $m$ emails, though this does not guarantee the inferred organization is a tree.
\citet{maiya2009inferring} use a distance-based approach that assumes interactions are most common between supervisors and their direct reports, as well as between organizational peers.
\citet{gupte2011finding} define ``agony" as the difference in rank plus one for communications that are directed from a lower-ranked to a higher-ranked employee, and infer the organizational hierarchy that minimizes agony.
While there are many hierarchical clustering methods that infer a hierarchy given a tree input, such as the hierarchical random graph model from \citet{clauset2008hierarchical}, their objective is different from the aforementioned methods.
Hierarchical clustering methods do not produce a spanning tree, i.e. a tree with a vertex set equal to that of the input graph.
The vertices of the output hierarchy are not the same as the vertices of the input network, hence hierarchical clustering is a different task than tree reconstruction.

Much of the methodology used to study the relationship between organizational hierarchy and communication patterns (and much of what is known about this relationship) comes from analyses of the Enron email corpus \citep{klimt2004enron}, which consists of email communications among a group of senior employees of the firm in advance of its collapse in 2001.
Appendix Table \ref{tab:enron} shows a summary of these methods.
Despite the repeated analysis of this data, the full organizational hierarchy of Enron employees is not publicly available: the Enron email corpus only consists of 158 employees that mostly belong to upper management.
More insidiously, insights gained from analyses of the Enron data may reflect features of this uniquely dysfunctional corporate environment: observed associations between the email network and organizational hierarchy in this failing business may not be representative of more successful organizations.
Indeed, \citet{eckhaus2018managerial} use sentiment analysis to show that a semantic measure of ``hubris'' in email content increases closer to the Enron collapse.
Similarly, \citet{diesner2005exploration} argue that during the crisis, communication transcended the formal positions more than in previous months.
Consequently, it is essential to move beyond the study of the Enron corpus in order to understand how an organization's structure relates to its communication dynamics.

Toward this end, researchers have conducted several empirical analyses of other organizations, but each is limited in scope.
\citet{holtzhausen2002effects} study a South African financial services organization to assess the relationship between decentralization and communication.
However, both communication and decentralization were only measured via survey response rather than via communication logs and the actual organizational hierarchy.
\citet{guimera2006real} characterize a university email network and relate it to the community structure within the institution.
Similarly, \citet{boeva2017analysis} study a telecommunications company to evaluate metrics for comparing the formal organization to the informal communication structure, and \citet{sims2014hierarchical} study communication patterns within Los Alamos National Laboratory
through email records.
\citet{edge2020workgroup} and \citet{athreya2022discovering} study email communication within the Microsoft organization; the former computes and visualizes network embeddings for work groups over time, and the latter uses similar network embeddings to detect COVID-19 as an exogenous shock to the organization.
Few studies provide a holistic view of the relationship between organizational structure and communication dynamics. In several cases the organizational hierarchy is coarsely measured and communication dynamics are only considered between groups consisting of many individuals (colleges and administrative blocks in \citet{guimera2006real}, divisions in \citet{boeva2017analysis}, and programs/projects in \citet{sims2014hierarchical}) rather than between individual employees.  Sometimes the organizational structure itself is inferred from the communication data (work groups are defined as groups of individuals with high within-team communication by \citet{edge2020workgroup} and \citet{athreya2022discovering}) implicitly assuming a positive relationship between hierarchy and communication and thus effectively using the same data as both input and response.

How do emergent communication dynamics relate to organizational hierarchy in a large and successful organization? 
Are communication network centrality measures a good proxy for organizational position?
What is the relationship between organizational rank and frequency of communication?
Empirically testing organizational theories, especially those postulating relationships between network centrality measures and organizational position, would require data from a large organization's individual-level organizational hierarchy with finely observed communication patterns between all employees.
Availability of replicates -- data from subgroups or teams within the organization -- would permit testing of organizational theories across business task areas. 

In this paper, we study the relationship between the formal organizational hierarchy and emergent email communication dynamics among 241,718 employees at Microsoft Corporation in May 2019. The data set is larger than any other organization described in the literature and consists of each employee's location in the company's supervisory/reporting hierarchy, as well as monthly email communication counts between all pairs of employees.
We characterize the topological distance and direction of email messages passed up, down, and across the organizational hierarchy.
By dividing the organization into 88 teams -- organizational trees rooted at the senior leadership level -- we identify distinct communication network structures within and between teams.
We also evaluate established methods for latent structure estimation and employee rank classification using the email communication network, and assess their correspondence to the true organizational tree across the whole organization and within individual teams.
We conclude with a discussion of the relationship between routine employee communication patterns and the supervisory hierarchy, and empirically evaluate several theories of organizational management and performance.


\section{Data}

Data used in this analysis consist of email communications between employees and the organizational hierarchy of all employees in Microsoft Corporation.
Email communications consist of monthly counts of emails between pairs of employees distinguished by sender and recipient.
A total of 95.5 million emails between employees are represented in these counts.
The organizational data are represented by directional reporting relationships between employees.
Unique anonymous employee identifiers link employees across both data sets.
A total of 241,718 unique employee identifiers are present across the data sets, with 89\% present in both data sets. 
We construct two networks by selecting records belonging to employees present in both data sets.
In the directed email communications network, nodes correspond to employees and edges are weighted by the total number emails over one month.
The weighted adjacency matrix $A$ consists of elements $A_{uv} = \#~\text{emails sent from employee}~u~\text{to}~v$.
The organizational tree is constructed using directional reporting relationships, where the Chief Executive Officer is the root.
Nodes in the organizational tree represent employees and edges represent (unweighted) reporting relationships. 

To accommodate the existence of divisions with distinct business roles within the company, we divide the organization into sub-organizations consisting of sub-trees of the organizational hierarchy, which we call \textit{teams}.
By defining teams, we are able to study patterns of within- and between-team communication.
Moreover, teams provide a notion of replicates that allow us to compute averages and examine trends between the organizational hierarchy and the communication network.
To construct teams, we first identify a set of \textit{team leaders}, defined as employees two steps below the CEO in the organizational chart.
A team's organizational tree is defined as the rooted tree with the team leader as the root, which is a a subtree of the entire organizational tree.
Its corresponding email network is the email communications network induced by all of the organizational descendants of the team leader; the team email network is a subgraph of the entire email network.
Because not all employees are present in both the email network and organizational tree, we take the intersection of the nodes in a team's organizational tree and email network and redefine these networks as the subgraphs induced by this intersection.
We exclude teams with fewer than 100 members from this analysis. 

As a final data cleaning step, we remove employees without any record of communication if they are terminal nodes in the hierarchy, thus preserving the tree structure.
This ensures that our distance measures between employees are meaningful and well-defined.
Since, by definition, terminal nodes in the hierarchy have no out-going edges, we can freely discard them without breaking the tree structure of the organization.
As these nodes have no communication sent or received, i.e. their total degree in the communication network is zero, they add nothing to the analysis. 
On the other hand, we keep nodes with no communication that are not terminal nodes, which implies they have at least one reporting relationship in the organizational hierarchy.
To ensure the communication network is connected, we add a single directed edge from the highest out-degree node in their team.
This addition is not critical to the overall analysis, but a connected network is necessary to compute some of the network summary statistics and centrality measures.
We chose the highest out-degree node in their team based on the idea that this person was the most likely person to send emails, and subsequently least likely to change the email graph topology.
In total, we add only 788 edges, which is less than $0.004\%$ of all of the edges.

Our final email network is the union of all of the teams' email networks plus all of the original edges between the different teams.
By construction, each team is connected in the organizational hierarchy, and the overall email network is also connected.
It consists of 197k nodes and 24.6 million edges.

For every node, we record several employee features.
First, we record each employee's \textit{supervisor}, which is the employee that the node reports to in the organizational tree.
We say that a group of teams is in the same \textit{division} when their respective team leaders share a supervisor, who by construction reports directly to the CEO.
Finally, we record each node's \textit{level}, which is the number of steps below the CEO in the organizational hierarchy.
In total, the data set consists of a total of 88 teams within 12 divisions.
The sizes of the teams range from 99 to 19,458 (average 2,237, standard deviation 3,612) and the \textit{depth} of the teams, which is the maximum level of individuals in the team, ranges from 5 to 13 (average 8.5, SD 1.7).

Figure~\ref{fig:ms_vis} shows a visualization of the organizational tree and email network using the Gephi software package \citep{bastian2009gephi}.
Teams are colored with the same color scheme in both networks, and colors are assigned to teams on a spectrum so that teams within the same division have similar colors.
On the left, black edges denote the top of the organizational tree, from the CEO to the division leaders to the team leaders.
Within teams, nodes are organized first by level and then by direct report.
On the right, edges are directed and colored according to the source of emails, i.e. the email sender.
For visual clarity, we only include edges where emails were reciprocated and at least 30 total were exchanged over the course of the month \citep{guimera2006real}.
We then visualize the giant connected component, which consists of 94\% of the original nodes (185,276) but just 3\% of the edges (762,936), thus making the visualization manageable to compute and clearer to render.

\begin{figure} 
    \centering
    \includegraphics[width=.495\textwidth]{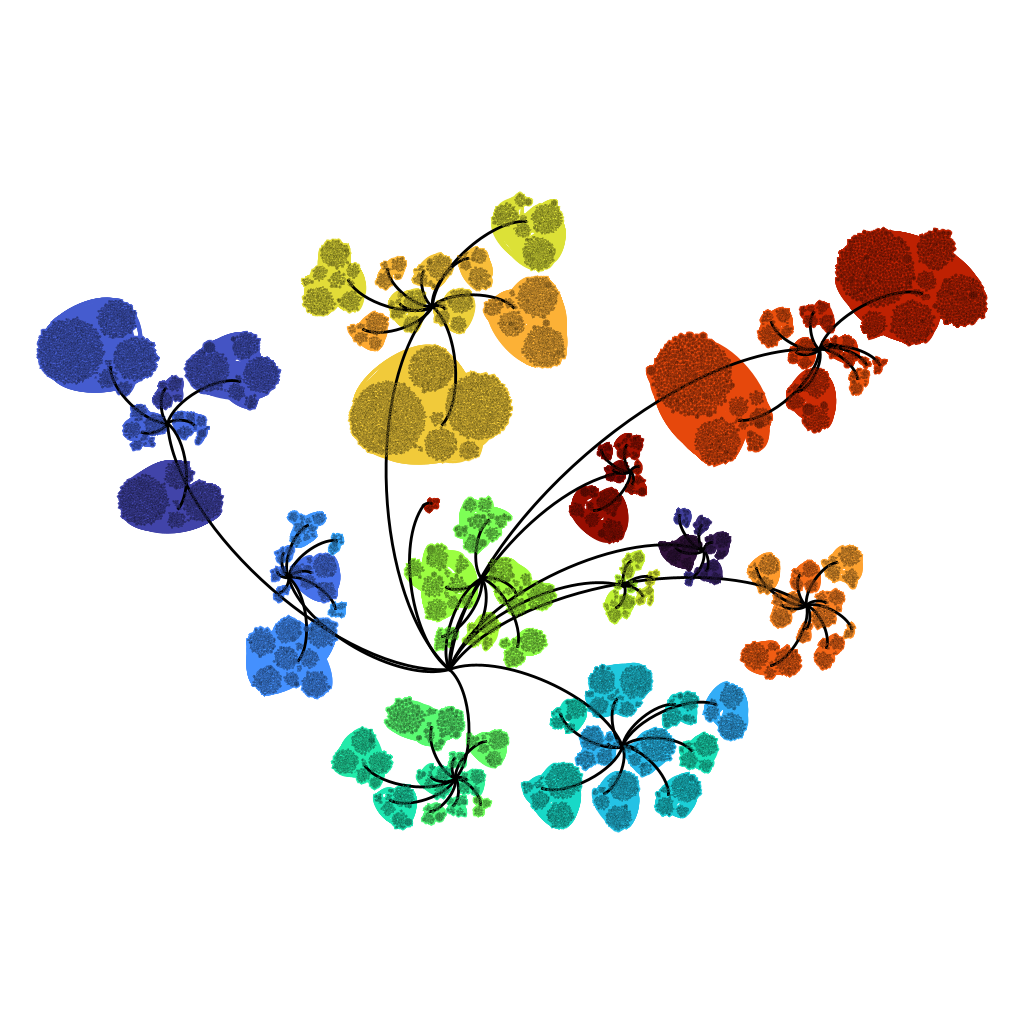}
    \includegraphics[width=.495\textwidth]{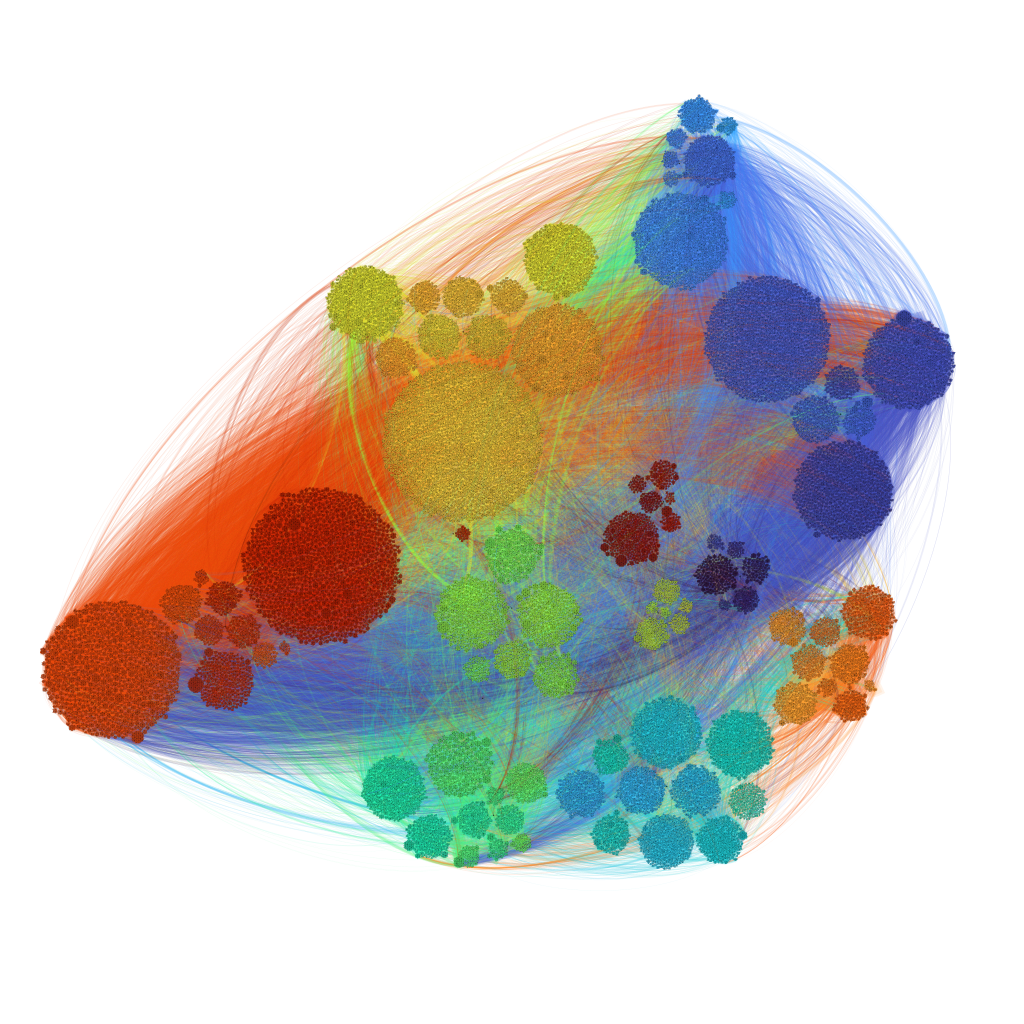}
    \caption{(Left) The organizational tree consisting of 196,832 employees in May 2019. The nodes are colored by team on a gradient such that teams in the same division are closer in color. The CEO is the black dot from which all of the black edges emanate. (Right) The email network consisting of all employees in May 2019 and the 24.6 million edges representing 84.7 million emails using the same color scheme for teams. Edges are directed and colored by source team. For visualization purposes, the email network has been coarsened as described in the text.}
    \label{fig:ms_vis}
\end{figure}


\section{EDA}\label{sec:eda}

\subsection{Email communication network} \label{sec:eda_ind}

We begin by studying the degree distributions of the email network.
The \textit{total degree} is the number of email relationships with nonzero weight, which we separate into \textit{in} (received) and \textit{out} (sent) degrees.
The average total degree is 250 (SD 402), and the median is 166 indicating a heavy right skew in the degree distribution.
The 95th percentiles of in and out degrees are 349 and 451, respectively.
By symmetry, since both sender and recipient are in the data set, the average in and out degrees are equal to half of the average total degree.
Figure \ref{fig:email_eda} shows a rapid decay in degree frequency that is consistent with patterns observed in other communication networks \citep{onnela2007structure}.
The long tail indicates that there are a few employees who receive emails from thousands of employees, and even fewer ``super-senders" who send emails to tens of thousands of employees.
This highly skewed out-degree distribution has been been identified in prior work on communication patterns in organizations, and may reflect activity related to email lists or listservs \citep{guimera2006real, onnela2007structure}.

Based on a Kolmogorov-Smirnov test, the total degrees follow a power-law distribution $(\propto x^{-\alpha})$ with an estimated power law exponent of $\hat{\alpha} = 3.8$ (SD 0.1) \citep{clauset2009power, JSSv064i02}.
However, we also include the recent scale-free analysis proposed in \citet{broido2019scale} at the team-level in Appendix \ref{sec:scale_free}, which shows that 69 of the 88 teams are not scale-free.
In and out degrees also show a strong positive Pearson correlation ($\rho = 0.66$). The same trend holds ($\rho = 0.71$) for \textit{total strength}, which is the sum of the edge weights and represents the total volume of email.
\textit{In} and \textit{out strength} are the total number of emails received and sent, respectively.
The average total strength is 861 (SD 1204), the median total strength is 563, and the 95th percentiles of in and out strength are 1,272 and 1,461, respectively.

The first question about the emergent communication dynamics is to what extent communication clusters around team.
One summary metric of this relationship is modularity, which measures the cluster strength of a network given a vector of cluster assignments for the node set \citep{newman2006modularity}.
Using teams and divisions, the email network has a weighted modularity of 0.64 and 0.61, respectively, similar to modularity measures reported in other communication networks in which teams are defined implicitly by maximizing modularity \citep{edge2020workgroup}.
To visualize this clustering, we show the team adjacency matrix $P$ in Figure~\ref{fig:email_density}.
The $(i, j)$ element is the proportion of pairs of employees from teams $i$ and $j$ that have ever exchanged an email: 
\begin{equation*}
    P_{ij} = \begin{cases}
        \Big(n_in_j\Big)^{-1} \sum\limits_{\substack{u \in \text{team}~i\\ v \in \text{team}~j}} 1(A_{uv} + A_{vu} > 0) & \quad \text{if}~i \neq j \\[1em]
        \displaystyle\binom{n_i}{2}^{-1} \sum\limits_{\substack{u,v\in \text{team}~i\\ u \neq v}} 1(A_{uv} + A_{vu} > 0) & \quad \text{if}~i = j
    \end{cases} \enskip ,
\end{equation*}
where $n_i$ is the number of employees in team $i$ and $1(\cdot)$ is the indicator function.
The rows and columns of the block matrix are sorted so that teams are clustered by their divisions.

Figure \ref{fig:email_density} shows that more emails are exchanged within teams than between teams, and the block diagonal structure also shows that communication is more common within divisions.
The average diagonal (within team) is 0.091, which is the average proportion of pairs in the same team that exchanged at least one email in May 2019, and the average off-diagonal (between team) is 0.002.
This indicates that teams are dense but communication between teams is sparse.
The last row/column represents the CEO and their direct reports, which is evidence of the ``rich club" phenomenon -- the tendency for central nodes to be embedded in dense communities -- at the top of the organizational chart \citep{colizza2006detecting, dong2015inferring}.

Figure~\ref{fig:email_eda} shows communication frequency within different groups of employees: all employees, those at the same level of the organizational hierarchy, those within the same division, those within the same team, and those who share a supervisor. 
Communication patterns largely cluster around the organizational structure: within teams, within divisions, within the same organizational level, and among those with the same supervisor. 

\begin{SCfigure} 
    \centering
    \includegraphics[width=0.6\textwidth]{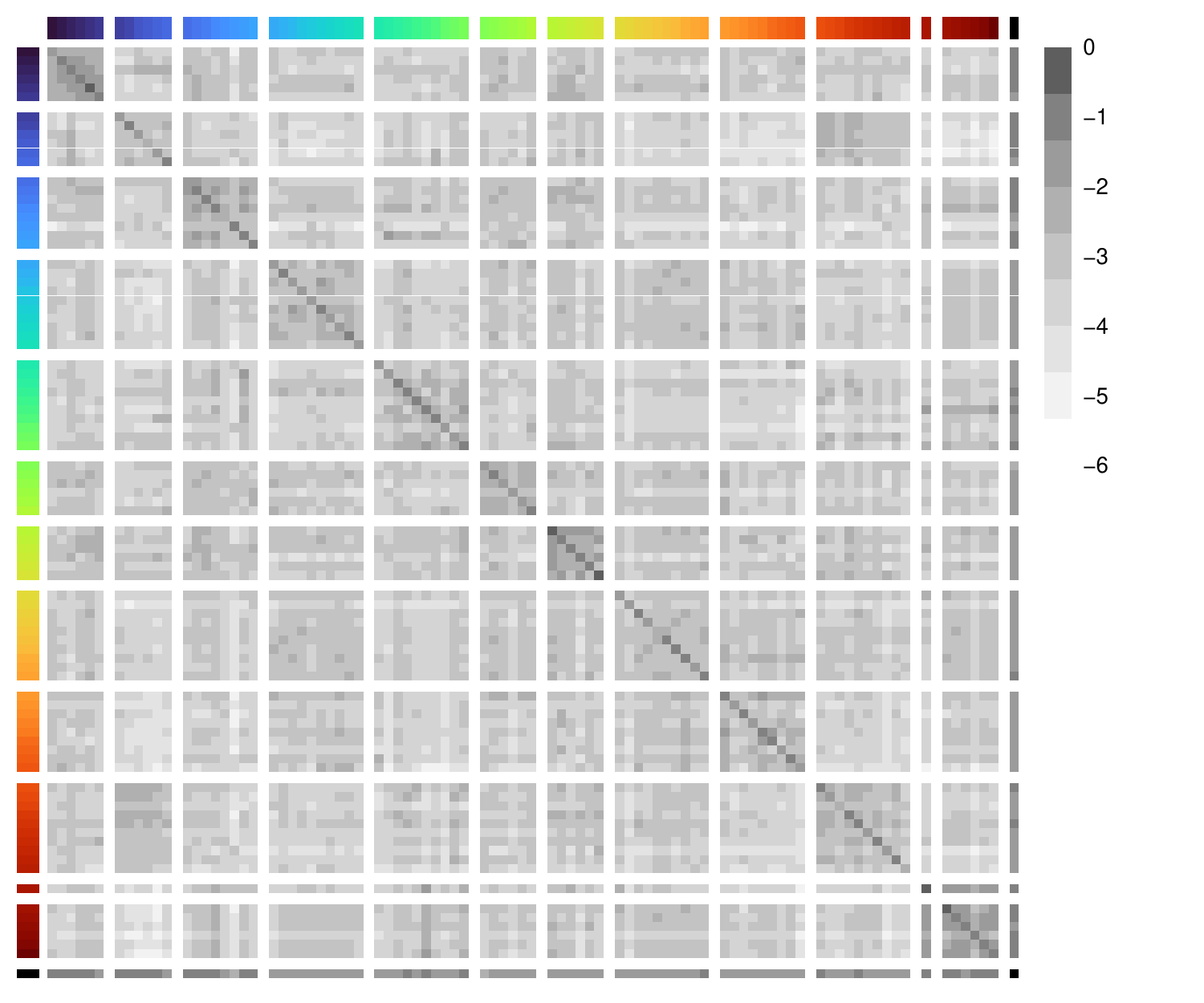}
    \caption{Email frequency (log10 scale) within and between teams. Darker colors indicate more frequent communication. Teams are indicated using the same colors as in Figure~\ref{fig:ms_vis} and are blocked according to adjacency in the organizational hierarchy. The diagonal represents within-team communication. The last row/column, colored in black, represents the CEO and their direct reports.}
    \label{fig:email_density}
\end{SCfigure}

\begin{figure} 
    \centering
    \includegraphics[width=.325\textwidth]{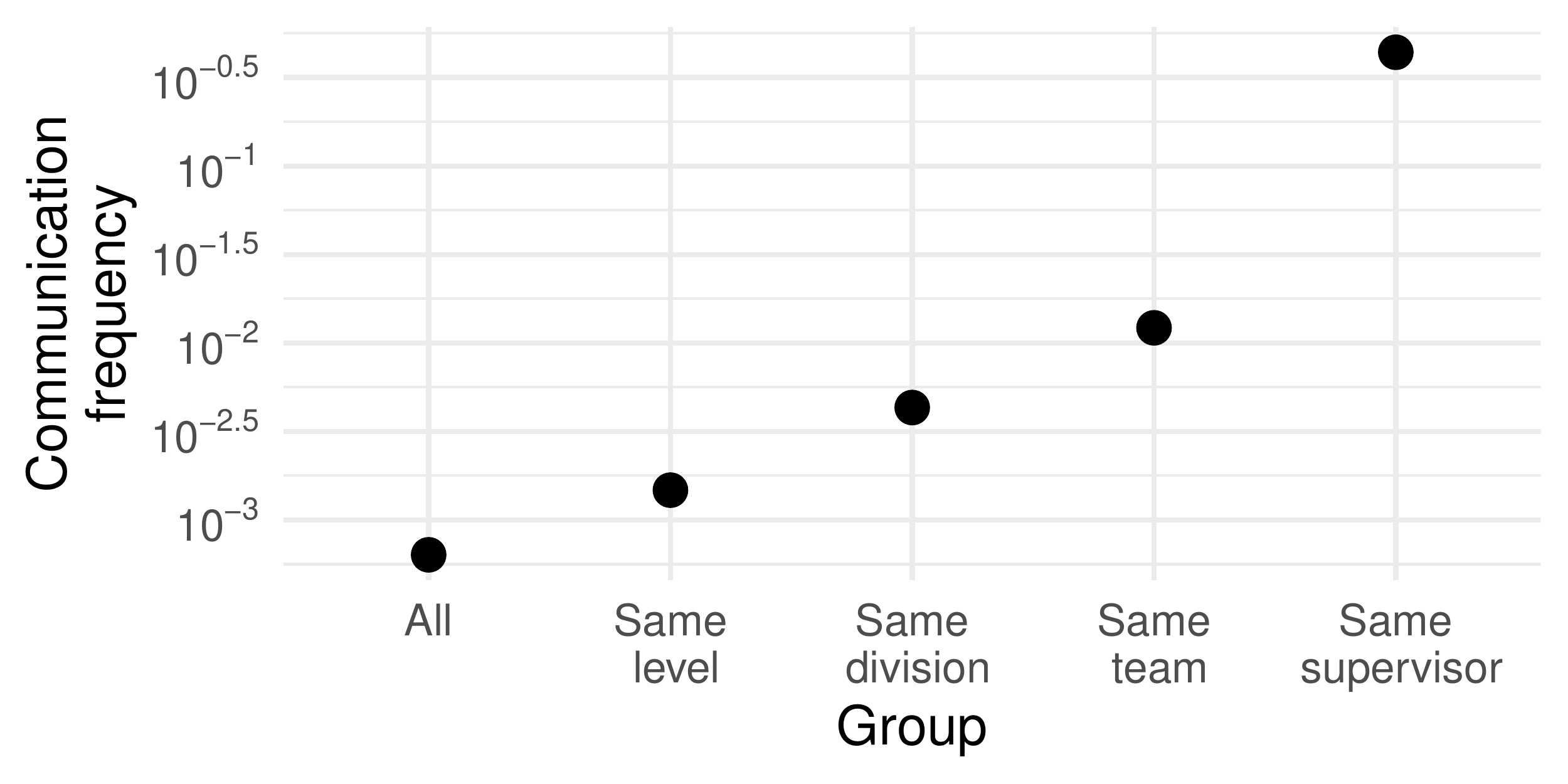}
    \includegraphics[width=.325\textwidth]{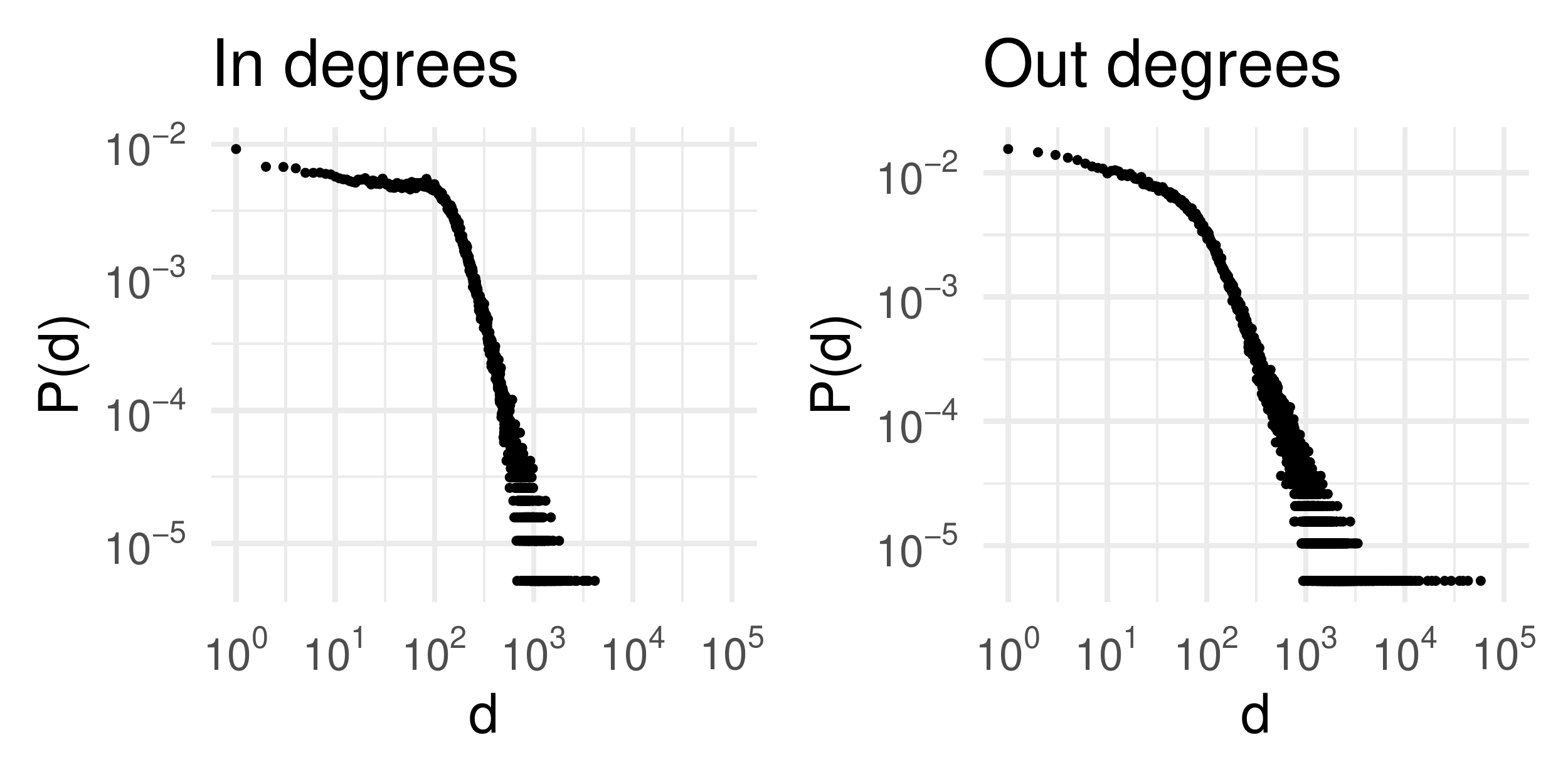}
    \includegraphics[width=.325\textwidth]{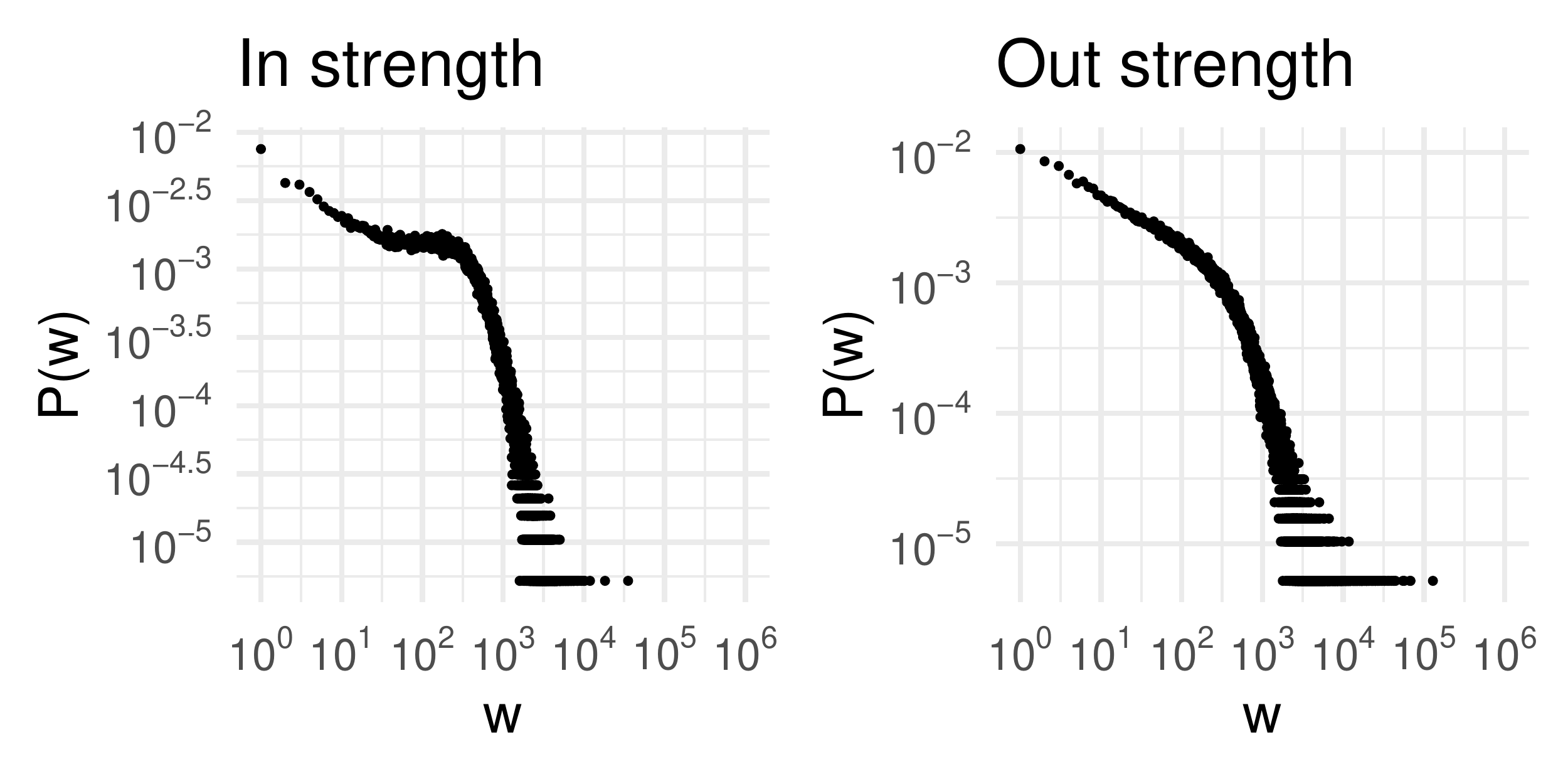}
    \includegraphics[width=.325\textwidth]{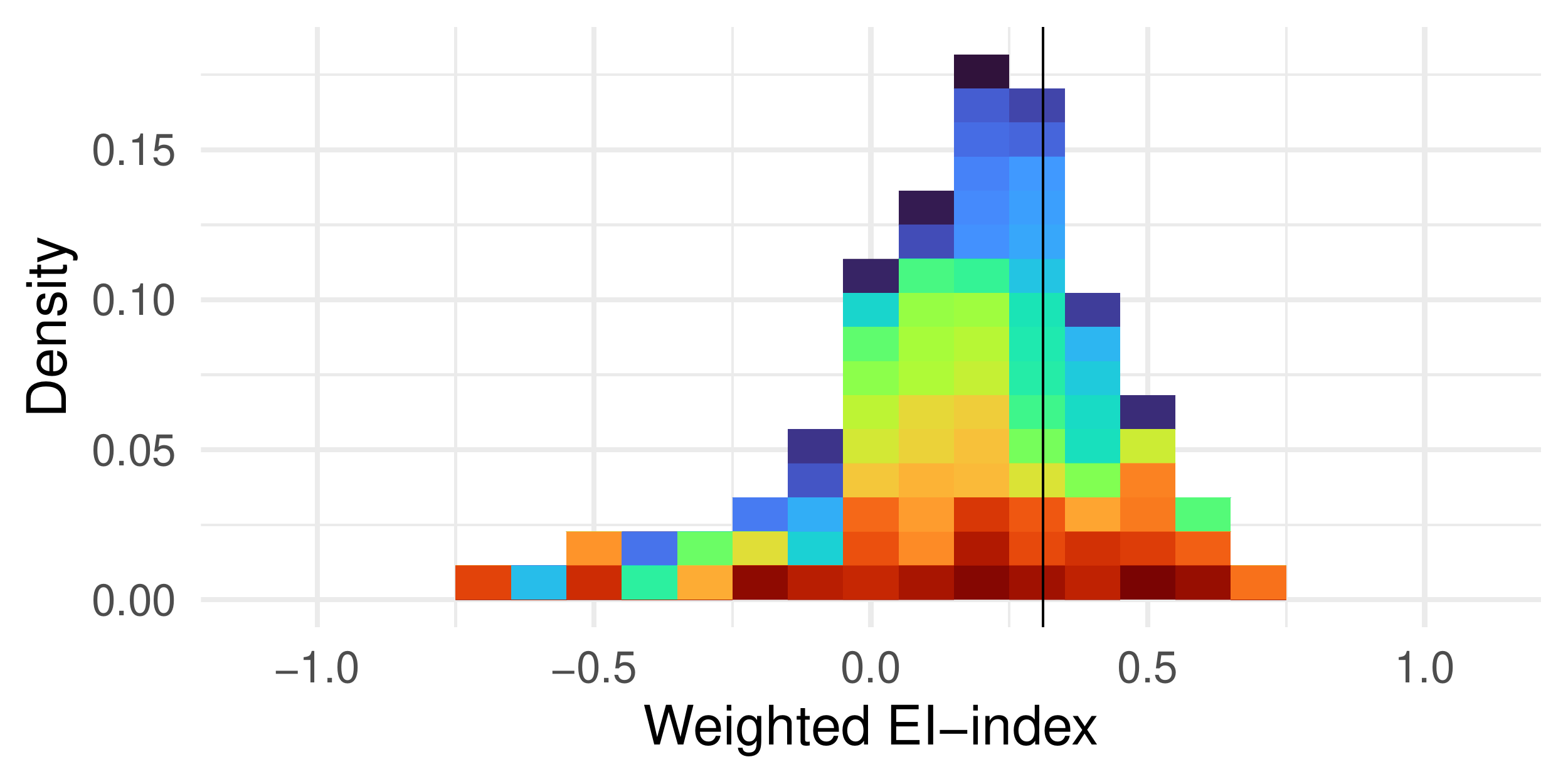}
    \includegraphics[width=.325\textwidth]{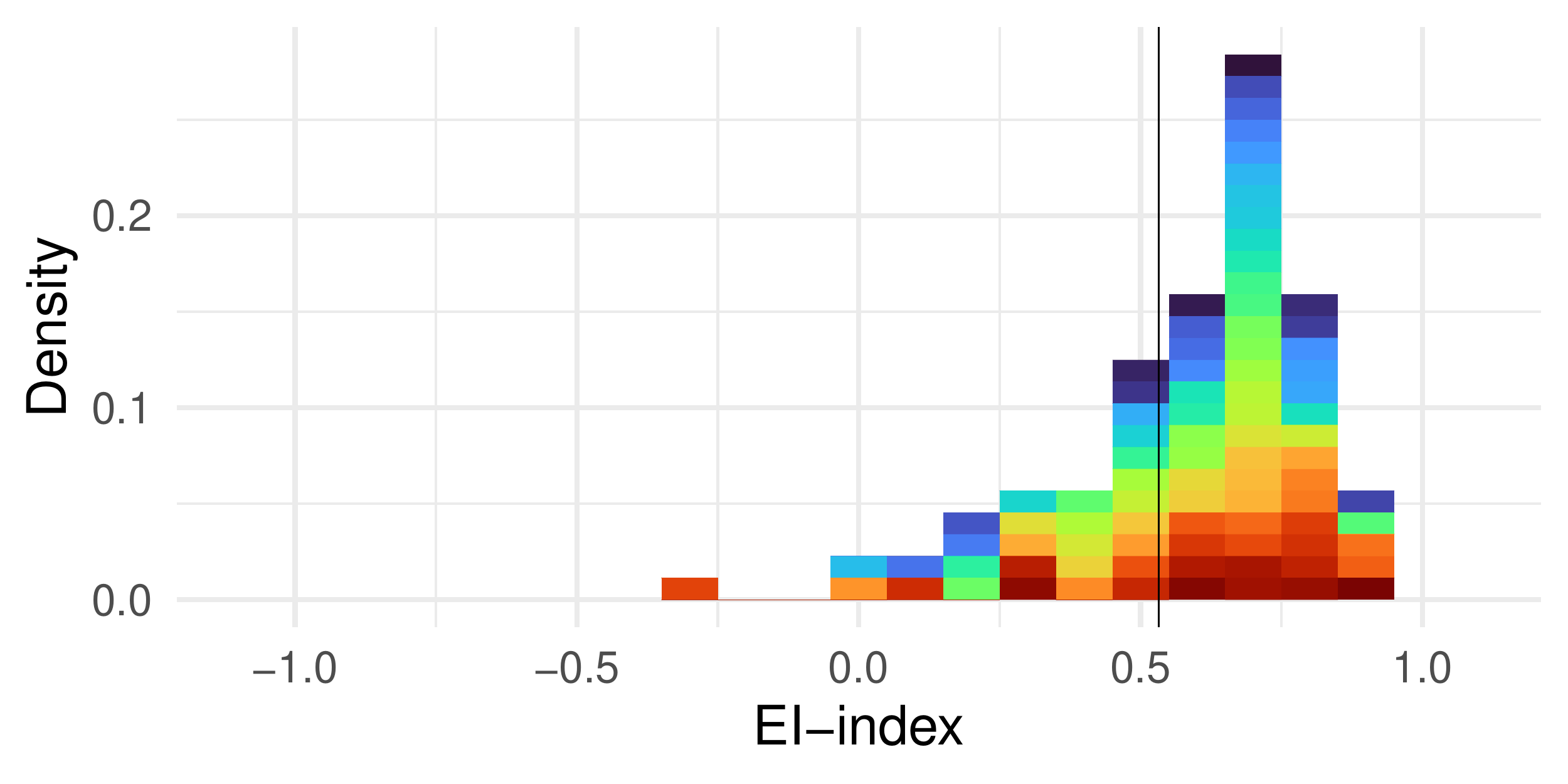}
    \includegraphics[width=.325\textwidth]{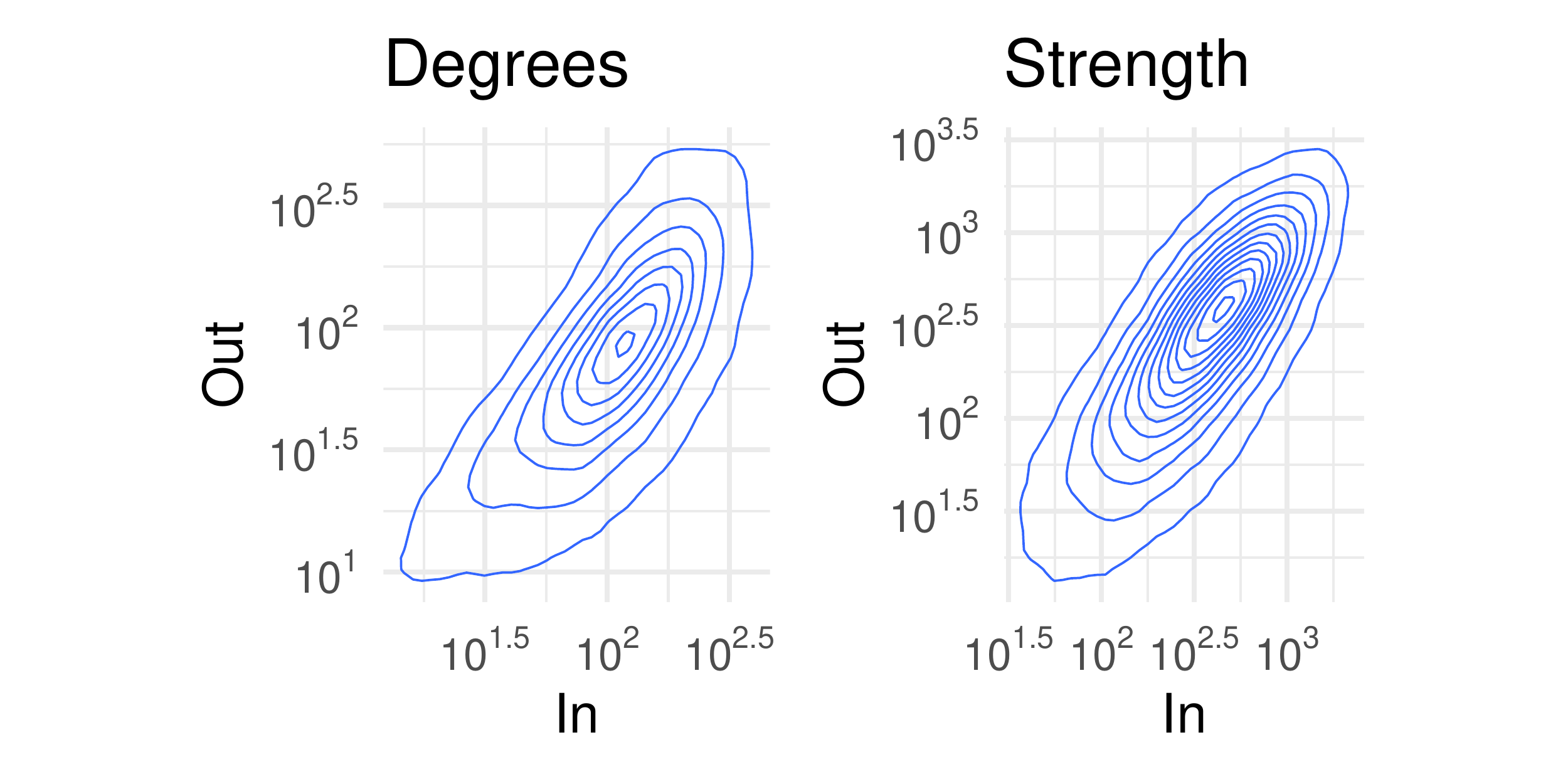}
    \caption{(Top) Left: Frequency of communication within different groups in the organization. Middle: In-degree and out-degree distributions. Right: In-strength and out-strength distributions. (Bottom) Right: Bivariate (in vs out) degree and strength distributions. Middle: Distribution of EI-index across (and colored by) teams, where the vertical line represents the entire organization's EI-index. (Left): Distribution of weighted EI-index across teams.}
    \label{fig:email_eda}
\end{figure}

\subsection{EI-Index}
\citet{krackhardt1988informal} argue that successful organizations respond better to crises when there are links between different teams, hypothesizing that organizations with a positive EI-index will respond effectively to crises.
The EI-index for team $i$ and the organization EI-index are given as
\begin{equation*}
    \text{EI-index}(i) = \frac{EL(i) - IL(i)}{EL(i) + IL(i)} \qquad \text{and} \qquad \text{EI-index} = \frac{\sum_i \big(EL(i) - IL(i)\big)}{\sum_i \big(EL(i) + IL(i)\big)} \enskip ,
\end{equation*}
respectively, where $EL(i)$ and $IL(i)$ are the number of external (between teams) and internal (within team) links, respectively, for team $i$:
\begin{equation*}
    EL(i) = \sum_{\substack{u \in \text{team}~i\\v \not\in \text{team}~i}} 1(A_{uv} + A_{vu} > 0) \qquad \text{and} \qquad
    IL(i) = \sum_{u\neq v \in \text{team}~i} 1(A_{uv} + A_{vu} > 0) \enskip .
\end{equation*}

Figure~\ref{fig:email_eda} shows the distribution of EI-index values across teams; all but one team has a positive EI-index (average 0.58, SD 0.23) with the EI-index of the organization as a whole equal to 0.53.
In contrast, the EI-indices for organizations reported by \citet{krackhardt1988informal} span the range from -1 to 1, but are only reported on networks with very few (25-36) nodes.
Though the number of communication links as a proportion of all possible links is higher with teams than between teams (see Figure~\ref{fig:email_density}), the unweighted EI-index is positive, indicating more external than internal links.
This happens because the number of external and internal links in the EI-index are not normalized by the number of possible external and internal links, and there are far more employees who are not a member of a given team (who could become external links) than there are employees who are members of that team (who could become internal links).
We also show in Figure~\ref{fig:email_eda} the distribution of weighted EI-indices in which $EL(i)$ and $IL(i)$ are replaced with the total volume of external and internal emails, respectively, i.e.
\begin{equation*}
    EL_W(i) = \sum_{\substack{u \in \text{team}~i\\v \not\in \text{team}~i}} A_{uv} + A_{vu} \qquad \text{and} \qquad
    IL_W(i) = \sum_{u\neq v \in \text{team}~i} A_{uv} + A_{vu} \enskip .
\end{equation*}
Here, we see that 21 of the 88 teams have negative weighted EI-index (average 0.15, SD 0.27).

\section{New Measures of Communication}

\subsection{Sent versus received emails}

Next, we investigate the dynamics of emails that are sent by employees compared to those that are received.
Link reciprocity for directed networks has been well studied in both binary \citep{garlaschelli2004patterns} and weighted settings \citep{squartini2013reciprocity}.
The reciprocity of our email network is $r = 0.31$ (SD = 0.01) without the weights and $r = 0.40$ (SD = 0.14) with the weights.
However, we are interested in reciprocity at the individual node level.
For this, we propose two measures of node-level reciprocity.
First, we measure the proportion of sent emails that are part of reciprocated communications, which we refer to as \textit{sent reciprocation}.
That is, the proportion of all email relationships in which $u$ sent at least one email to $v$ among all recipients $v$ who ever sent an email to $u$.
We similarly measure the proportion of received emails, or \textit{received reciprocation}.
For any individual $u$, these are defined as
\begin{equation*}
    \text{SR}(u) = \frac{\sum_{v} 1(A_{uv}>0) 1(A_{vu} > 0)}{\sum_{v} 1(A_{uv} > 0)} \qquad~\text{and}~\qquad \text{RR}(u) = \frac{\sum_{v} 1(A_{vu} > 0) 1(A_{uv} > 0)}{\sum_{v} 1(A_{vu} > 0)} \enskip ,
\end{equation*}
respectively. Figure~\ref{fig:eda_position} shows the difference of these proportions as a function of the employee's level, which we recall is the number of steps below the CEO in the organization, relative to the depth of that employee's team.
This normalization ensures that employees that are at the bottom of their team's hierarchy are grouped together no matter how deep each team is.
We see that employees higher in the organization send emails that are reciprocated more often than they reciprocate emails that they receive.

We also study the dynamics of sent and received emails as a function of position within one's team.
For every individual $u$ in team $i$, the \textit{hierarchical position} \citep{boeva2017analysis} is
\begin{equation*}
    \text{HP}(u) = \frac{\sum_{v \neq u} D_{uv}}{n_i - 1} \enskip ,
\end{equation*}
where $D_{uv}$ is the \textit{hierarchical difference} between $u$ and $v \in \text{team}~i$,
\begin{equation*}
    D_{uv} = \begin{cases}
        +1 & \text{if}~u~\text{is higher than}~v~\text{in the hierarchy} \\
        \ \ 0 & \text{if}~u~\text{is on the same level as}~v~\text{in the hierarchy} \\
        -1 & \text{if}~u~\text{is lower than}~v~\text{in the hierarchy}
    \end{cases} \enskip .
\end{equation*}

Inspired by this differencing, we define \textit{sent position} and \textit{received position} for individual $u$ as
\begin{equation*}
    \text{SP}(u) = \frac{\sum_{v \neq u} 1(A_{uv} > 0) \cdot D_{uv}}{\sum_v 1(A_{uv} > 0)} \qquad~\text{and}~\qquad \text{RP}(u) = \frac{\sum_{v \neq u} 1(A_{vu} > 0) \cdot D_{vu}}{\sum_v 1(A_{vu} > 0)} \enskip ,
\end{equation*}
respectively.
Thus +1 indicates every email was sent (received) to someone below in the organization, -1 to someone above in the organization, and 0 proportionately above and below.
Note that each sum is normalized by the volume sent by the individual.
The difference of these proportions is shown in Figure~\ref{fig:eda_position} as a function of HP.
The lower an employee is in the organization (the more negative their HP), their sent position is, on average, higher than their received position.
This is not surprising since there is a strong negative correlation $(\rho = -0.87)$ between SP and RP, which implies that the difference is going to be more positive the higher you are in the organization.
In short, if an employee sends most emails down the organization, then most emails they receive are likely to be sent from below, and vice-versa.

Finally, we evaluate the claim that the most important members in the organization -- those closest to the CEO -- have more communication connections \citep{cross2002making, nielsen2016work}.
Figure~\ref{fig:eda_position} shows that there is some truth to this claim: on average, both total degree and total strength in the entire email network decrease for employees further down the organizational chart.
We show in Appendix Table~\ref{tab:importance} that, assuming a linear model, the slopes of these trends are significantly less than zero.
By looking at individual teams, we can assess importance with several additional measures of node centrality that we cannot otherwise, for computational reasons, compute on the entire email network.
We compute four measures of centrality -- betweenness, closeness, eigenvalue, and authority \citep[see][for definitions]{kolaczyk2014statistical} -- shown in Appendix Figure~\ref{fig:eda_centrality}.

\begin{figure}
    \centering
    \includegraphics[width=.49\textwidth]{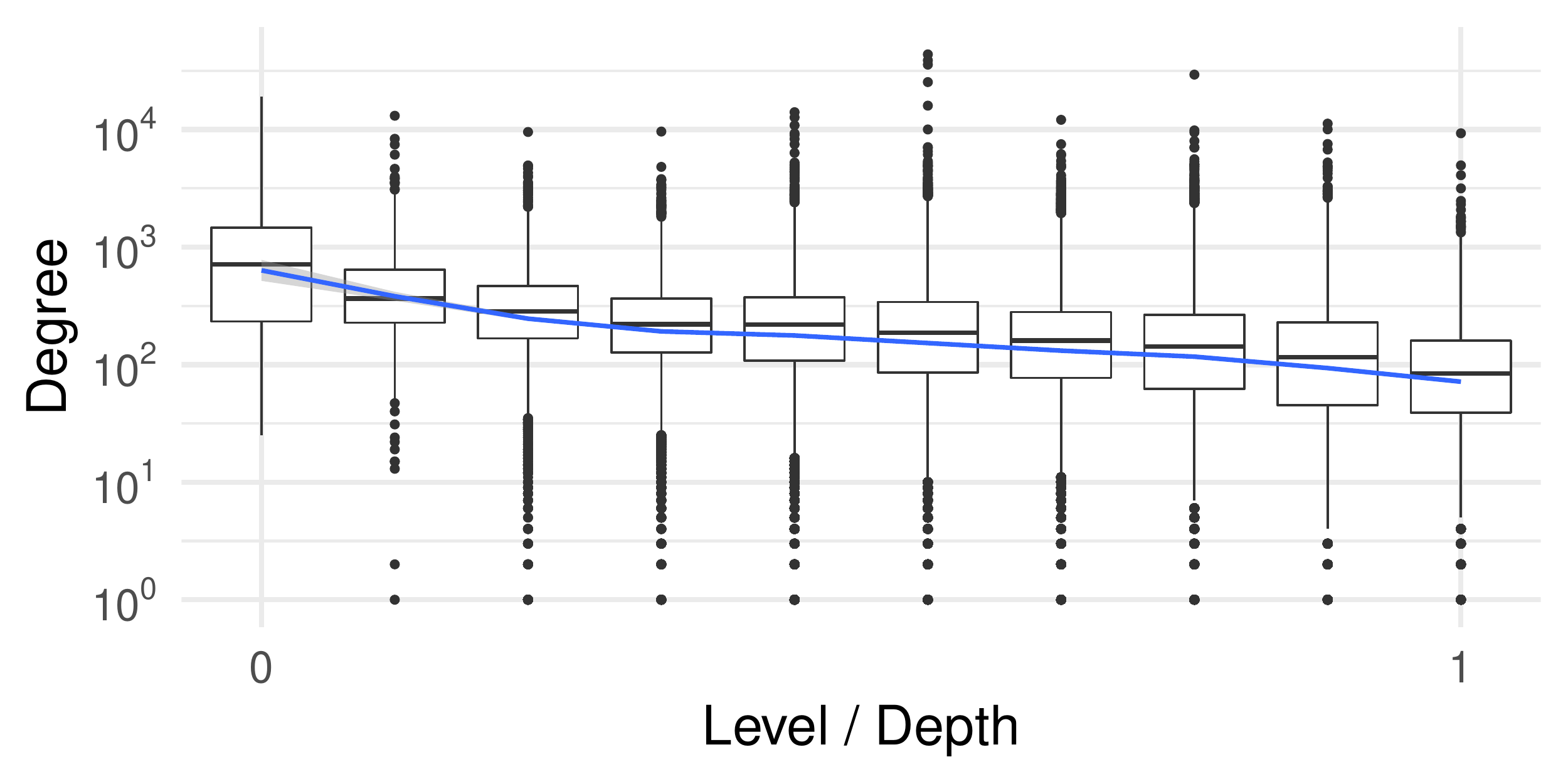}
    \includegraphics[width=.49\textwidth]{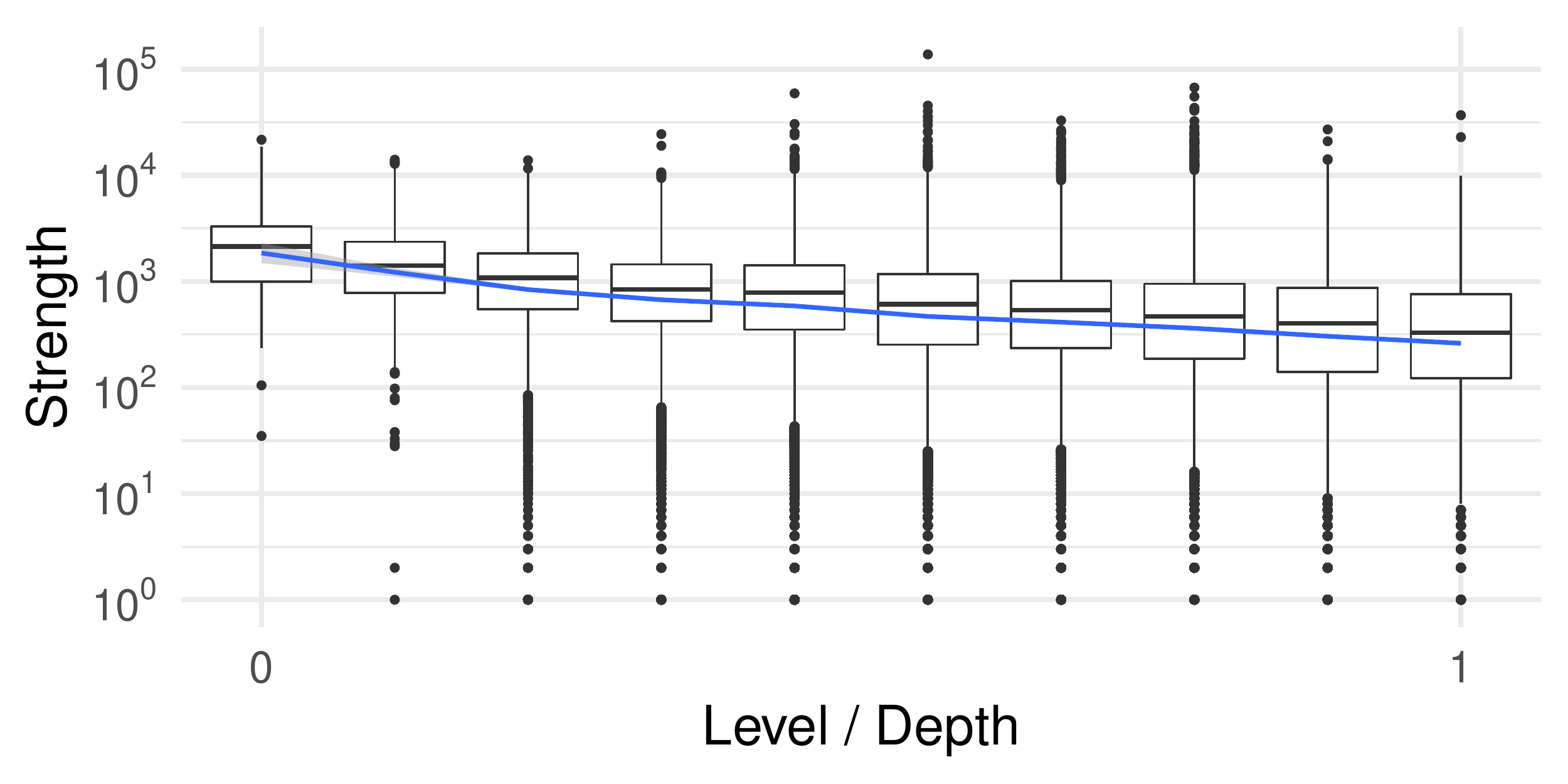}
    
    \includegraphics[width=.49\textwidth]{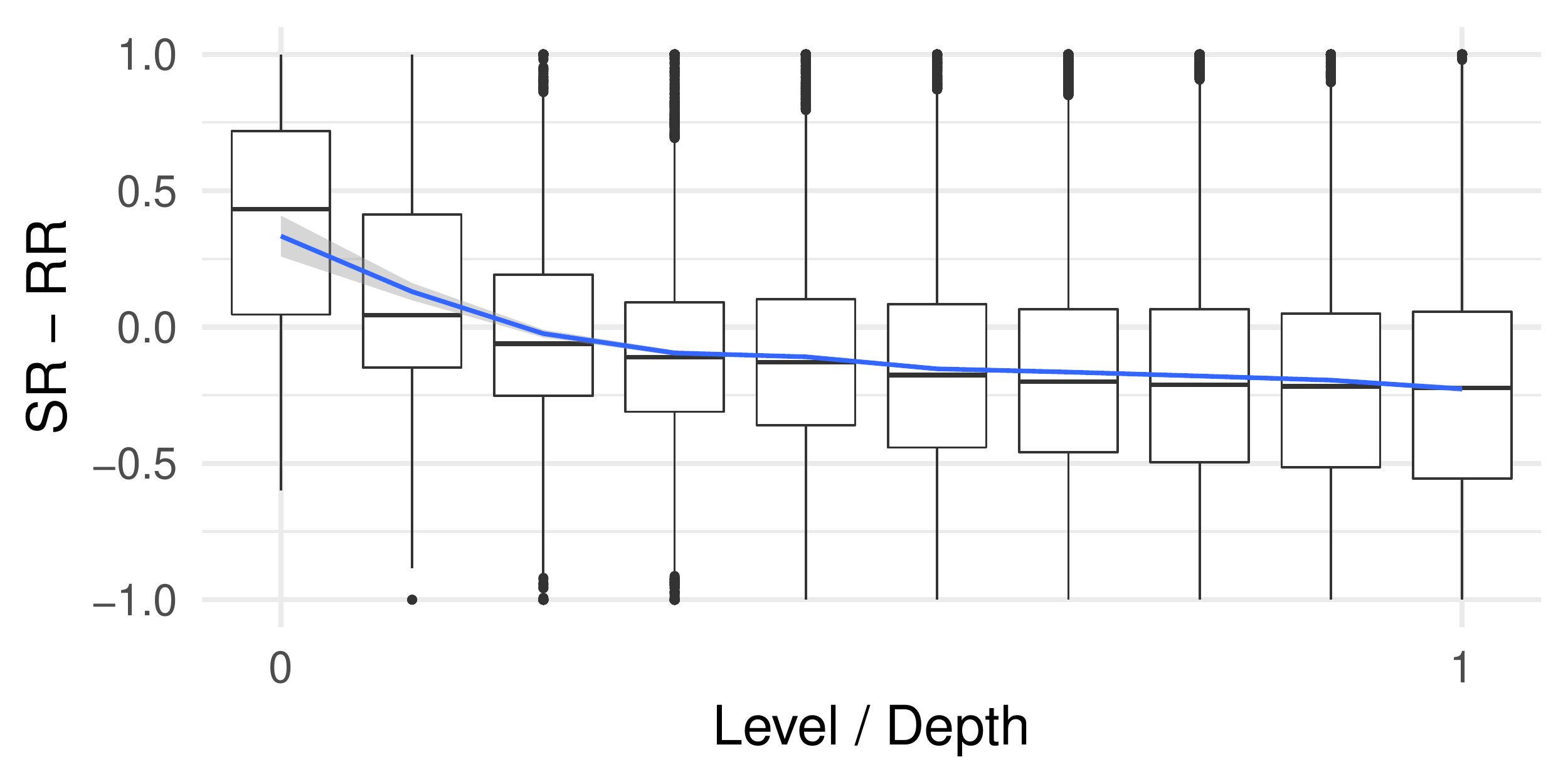}
    \includegraphics[width=.49\textwidth]{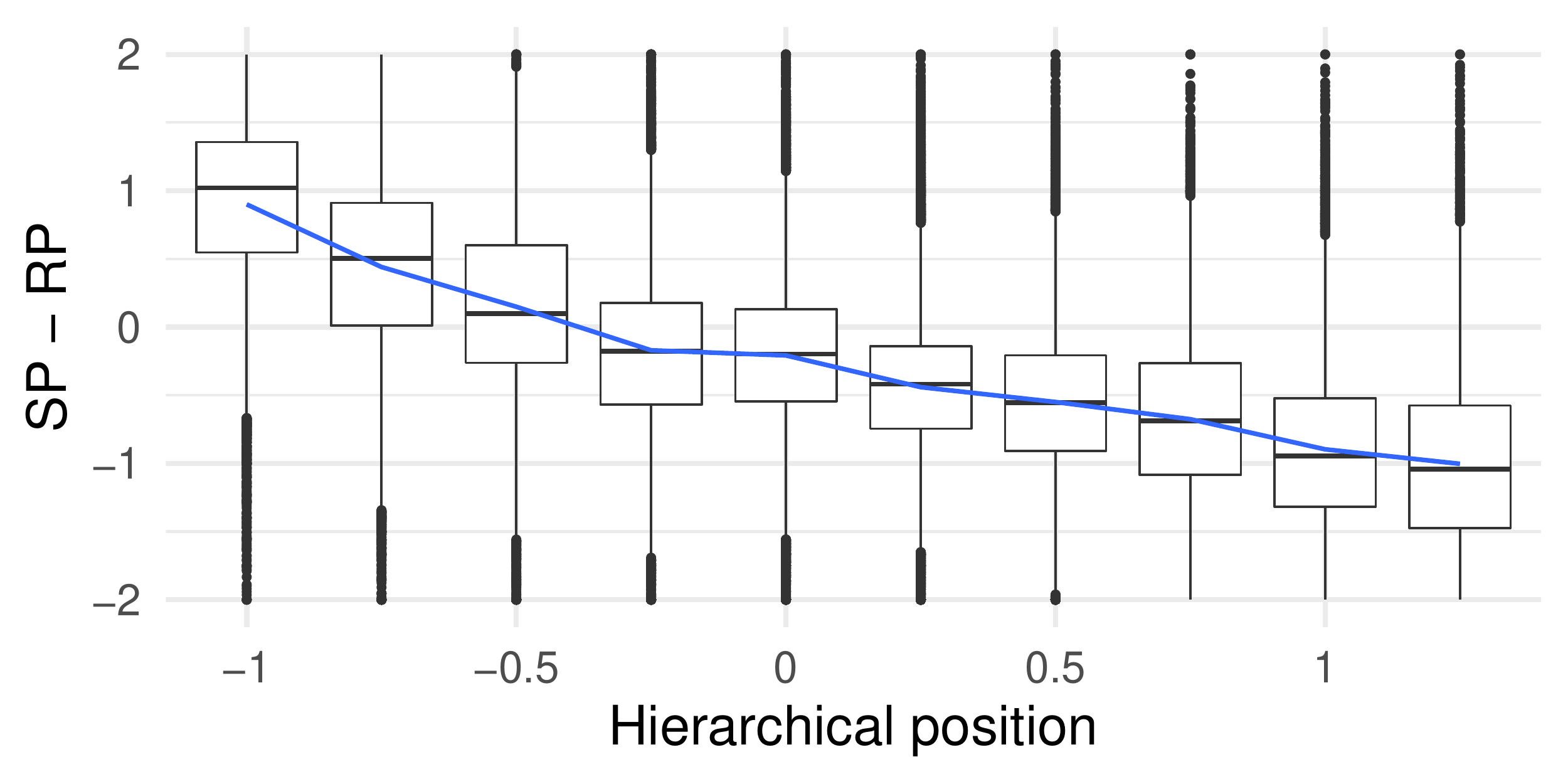}

    \caption{Measures of team-level email communication and reciprocity by organizational position. (Top) Total degree and strength in the entire email network by relative position in the organization. (Bottom) Left: Difference in the proportions SR and RR by relative position in the organization hierarchy. Right: Difference in the signed proportions SP and RP by HP. In all of the plots, we bin position into 10 equal groupings and we show the box plot within each bin along with a smoothed curve (blue) of the individual (non-binned) data.}
    \label{fig:eda_position}
\end{figure}

\subsection{Reporting distance}\label{sec:path_analysis}

Another approach to studying the relationship between organizational structure and frequency of communication is to consider the distance between two employees in the organizational chart.
A path between individuals $u$ and $v$ is a sequence of edges connecting them and the path length is defined as the number of edges in a path.
For each team, we compute three different path-length measures between each pair of employees.

For any two employees $u$ and $v$ in the organizational chart, we count the number of steps up $n^\text{up}(u, v) > 0$ and the number of steps down, $n^\text{down}(v, u) > 0$ in the shortest path from $u$ to $v$.
We describe how to compute these efficiently in Appendix~\ref{sec:n_up_down}.
Note that by symmetry, $n^\text{up}(u, v) = n^\text{down}(v, u)$.

We define three notions of reporting distance as a function of these path-length quantities: \textit{reporting distance}, \textit{signed reporting distance}, and \textit{directed reporting distance}.
Respectively, these distances between $u$ and $v$ are defined as
\begin{align*}
    \text{RD}(u, v) &= n^\text{up}(u, v) + n^\text{down}(u, v) \enskip , \\
    \text{SRD}(u, v) &= n^\text{up}(u, v) - n^\text{down}(u, v) \enskip , \\
    \text{DRD}(u, v) &= \text{RD}(u,v) \cdot \text{sgn}\big(\text{SRD}(u, v)\big) \enskip .
\end{align*}
RD is just the ordinary shortest path length.
SRD is a measure of the total distance travelled up or down the organizational tree to get from $u$ to $v$, which is zero if $u$ and $v$ have the same level, is positive if $u$ is lower than $v$ in the organization, and is negative otherwise.
SRD is close to the definition of ``agony" described in \citet{gupte2011finding}: $\text{Agony}(u,v) = \max\{\text{SRD}(v,u) + 1, 0\}$.
We may also regard SRD as a generalization of hierarchical difference \citep{boeva2017analysis}, since HD is just the sign of SRD.
In Appendix Figure~\ref{fig:hp_srd}, we present sent and received position, defined in terms of SRD, as a function of hierarchical difference generalized with SRD.
Finally, DRD is the reporting distance signed by whether $u$ is higher or lower than $v$ in the organization.

Note that only RD is a true distance; both SRD and DRD can be negative, and neither satisfies symmetry or the identity of indiscernibles, but we prove in Appendix~\ref{sec:rep_dist} Proposition~\ref{prop:tri_eq} that SRD satisfies the equality $\text{SRD}(u, v) = \text{SRD}(u, w) + \text{SRD}(w,v)$ for all $w$ in the tree, and hence SRD is a quasipseudometric on the set of employees \citep{kim1968pseudo}.

Figure~\ref{fig:eda_path} shows the average number of emails exchanged between pairs of employees as a function of these reporting distances. 
Appendix~Figure \ref{fig:eda_path_unweighted} shows analogous plots by communication frequency.
On average, the closer two employees are in the organizational tree, the more they communicate.
This shows that communication increases exponentially by local proximity to others, similar to the the global clustering in Figure~\ref{fig:email_density}.

\begin{figure} 
    \centering
    \includegraphics[width = \textwidth]{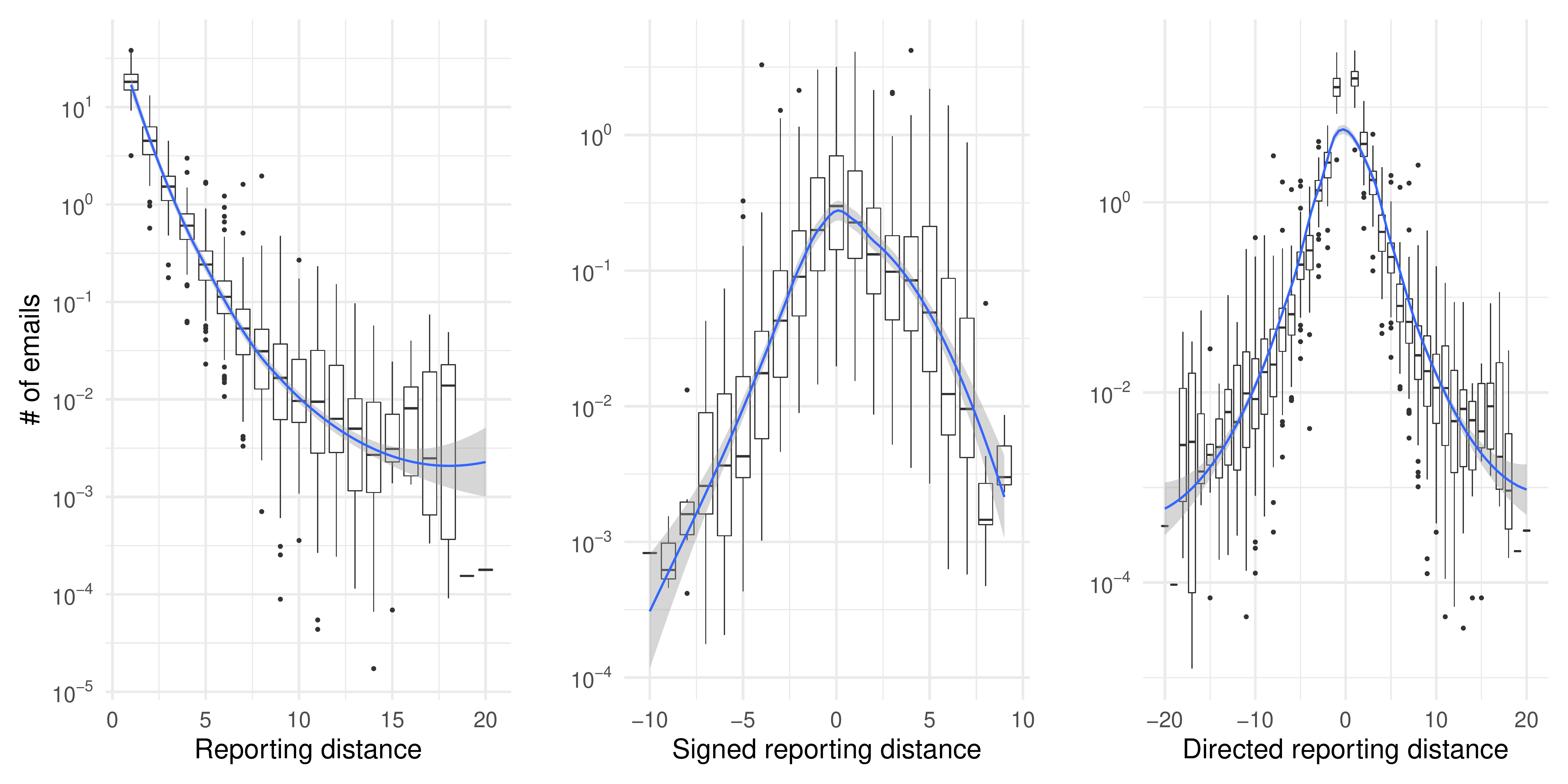}
    \caption{Pairwise reporting distances in the organizational tree and the average number of emails among all pairs in that reporting distance. Reporting distances are computed within each team and the box plots summarize the results across all of the teams. The individual team plots are shown in Appendix Figure~\ref{fig:eda_path_team}.}
    \label{fig:eda_path}
\end{figure}

Does the frequency of communication between employees depend on their relative ranks in the organizational hierarchy?
To assess whether the right and center plots of Figure~\ref{fig:eda_path} are symmetric about zero, we perform two permutation tests.
We define the sets
    \begin{equation*}
        S_k = \{(u,v) : \text{d}(u, v) = k\} \enskip ,
    \end{equation*}
which are all the pairs $(u, v)$ whose distance is $k$.
This allows us to define and compute the following test statistic:
    \begin{equation*}
        t(A) = \sum_{k=1}^{k_{\max}} \Big(\frac{1}{\vert S_k \vert}\sum_{(u, v) \in S_k}A_{uv} - \frac{1}{\vert S_{-k} \vert}\sum_{(u, v) \in S_{-k}}A_{uv}\Big)^2  \enskip ,
    \end{equation*}
where $k_{\max}$ is the maximum absolute distance observed.
If the relationship between distance in the organizational hierarchy and communication is symmetric around 0, then the pairs in $S_k$ and $S_{-k}$ should be exchangeable with respect to their communication levels.
This motivates our permutation test in which we obtain a null distribution by permuting the rows/columns of $A$ within the set $S_k \cup S_{-k}$.
This can be interpreted as randomly permuting the number of emails exchanged among pairs whose distances in the organizational hierarchy are $|k|$.
For each replicate, we obtain a new email network $\tilde{A}$ from which we can compute $t(\tilde{A})$.
By design, the relationship between the distance in the organizational hierarchy and communication in $\tilde{A}$ is symmetric around 0 under this permutation.

We repeat this $500$ times within each team, except for the 4 teams that had $> 10,000$ members, whose permutations were too computationally expensive.
We compare the empirical distribution of the test statistic under permutation $t(\tilde{A})$ with $t(A)$, shown in Appendix Figure~\ref{fig:eda_path_perm}.
We find that directed reporting distance is not symmetric about zero: that is, we reject the null hypothesis of symmetry under this permutation distribution at the empirical 95\% confidence level.
On the other hand, for signed reporting distance, we find the observed result is not significantly different than what we would expect for a truly symmetric relationship.
The empirical distribution is shown in Appendix Figure~\ref{fig:eda_path_perm}.

This implies that, with respect to SRD, communication up and down the organizational hierarchy is symmetric, whereas it is asymmetric with respect to DRD.
This is important because Agony \cite{gupte2011finding} can be defined in terms of SRD, and under the theory that ``higher rank nodes are less likely to connect to lower rank nodes,” tree reconstruction, which we explore in the next section, is performed by minimizing agony.
This is empirically contradicted based on our finding that SRD \textit{is} symmetric.
This finding both refutes the method of minimizing agony, and suggests an improvement, namely redefining Agony in terms of DRD.

\section{Reconstructing the organizational hierarchy}\label{sec:model}

Finally, we assess the existing methodology for reconstructing a tree structure from network data.
Each of the methods makes implicit assumptions about the relationship between the organizational hierarchy and emergent communication dynamics.
\citet{maiya2009inferring} propose a distance-based tree reconstruction model: ``as the distance between individuals within a hierarchy grows, we assume the probability of interaction decays."
\citet{gupte2011finding} propose a tree reconstruction method that minimizes agony in the communication network, based on the idea that ``when people connect to other people who are lower in the hierarchy, this causes them social agony" and thus ``higher rank nodes are less likely to connect to lower rank nodes".
As a benchmark method for tree reconstruction, we compute the minimum spanning tree of the communication network \citep{prim1957shortest}.
The minimum spanning tree of a graph is a tree with the same vertex set that minimizes the sum of the edge weights.
The implicit model for the minimum spanning tree is therefore one that minimizes the total emails exchanged between connected nodes.
As this is contrary to what we find -- the highest rates of communication are between connected nodes (Figure~\ref{fig:eda_path}) -- we also implement the maximum spanning tree.
We refer to these methods as ``Distance," ``Agony," ``Min ST," and ``Max ST," respectively.

To evaluate these approaches, we apply each reconstruction method separately to the team communication networks except for the 16 teams with $>$ 3,000 members.
For Agony, we implement the algorithm from \citet{tatti2014faster}.
For each reconstruction method, we compute the distance from the output tree to the true team organizational structure using the Frobenius distance and the centrality distance \citep{donnat2018tracking}.
For the definitions and interpretations of these tree distances, see Appendix~\ref{sec:tree_dist}.
We also measure the classification rate of the ``manager status" implicit in the reconstruction method, in which a node is a ``manager" if its in degree is positive, i.e. the proportion of nodes whose leaf status in the reconstruction agrees with the true hierarchy.

We note that Agony guarantees the output is a Eulerian graph, which means that the in and out degrees are equal, but this ensures that the graph is not a tree.
Similarly, Distance is guaranteed to produce a spanning graph, but a tree is not guaranteed since the maximum likelihood solution can place a node as the child of an unobserved root.
We find, in fact, that using the email totals as weights in the distance-based model results in this trivial solution (all nodes are children of an unobserved root) for all of the teams, and consequently we only show the results based on the unweighted communication networks.

\begin{SCfigure} 
    \centering
    \includegraphics[width=.65\textwidth]{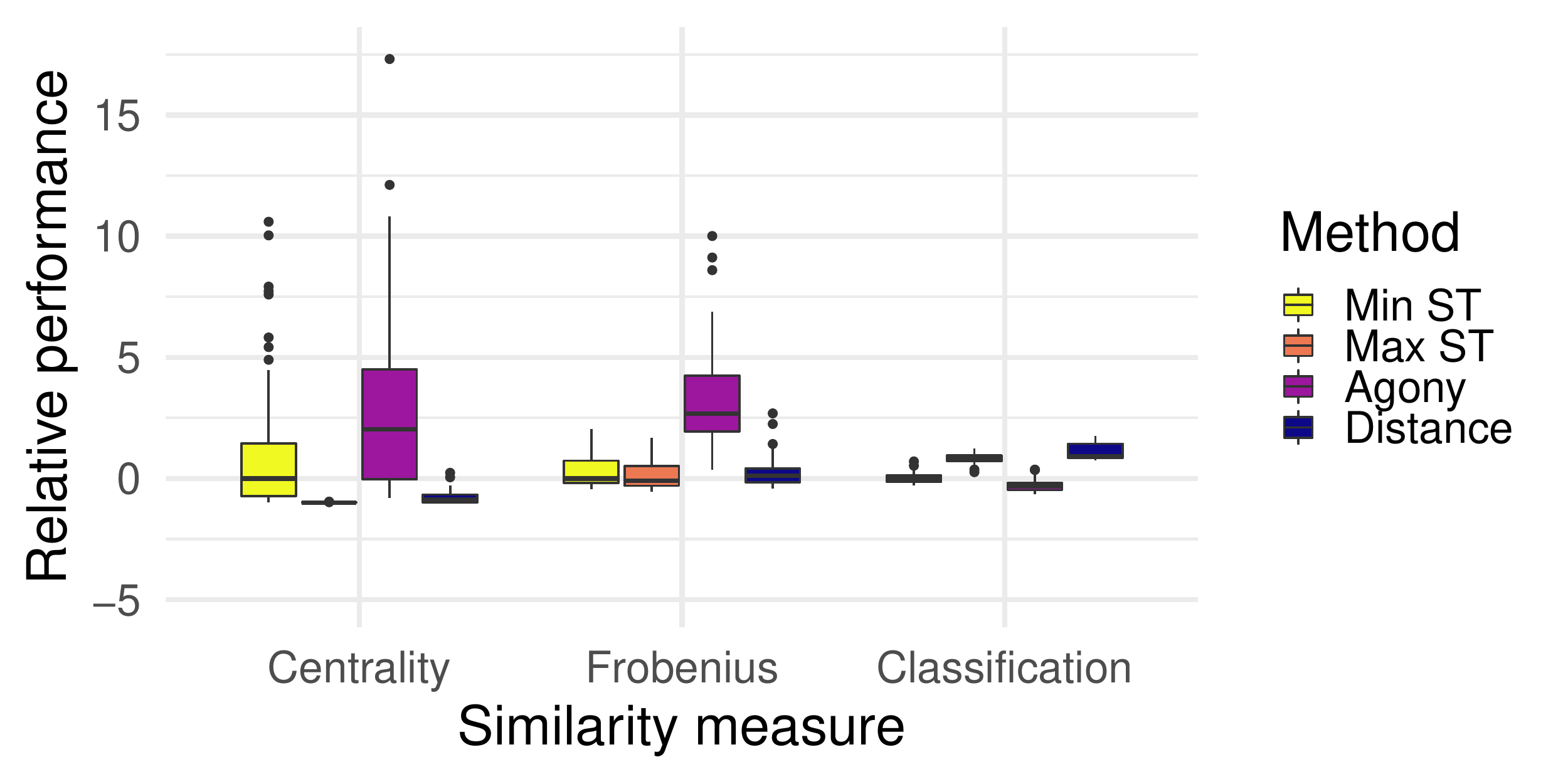}
    \caption{Performance of methods for reconstructing the organizational trees from the communication networks -- Min and Max ST \citep{prim1957shortest}, Distance \citep{maiya2009inferring}, and Agony \citep{gupte2011finding} -- relative to the median distance of Min ST. Negative values of Centrality and Frobenius, and higher values of Classification, indicate estimates closer to the true structures.}
    \label{fig:tree}
\end{SCfigure}

The results are given in Figure~\ref{fig:tree} and show that neither Agony nor Distance perform consistently better at reconstructing features of the organizational hierarchy than the baseline.
Instead, MaxST, which is not explicitly based on a theory of communication, performs the best.
We include the absolute performance of each method in Appendix Figure~\ref{fig:tree_abs}.
While the Frobenius and centrality metrics are hard to interpret, we can see that the classification accuracy is reasonable for both Max ST and Agony.
We also provide an analysis of the average difference in predicted level as a function of level down from the root.
This serves as a goodness of fit test restricted to increasing ranges of levels.
Appendix Figure~\ref{fig:tree_levels} shows that Min ST and Max ST are both better at recovering the top of the hierarchy than the lower levels.

To understand why these methods are lacking, recall that Agony is based on the hypothesis that higher ranking nodes are less likely to communicate with lower ranking nodes.
Under this implicit theory of communication dynamics, communication frequency ought to be asymmetrically and linearly decreasing around zero agony.
However, we find in Section~\ref{sec:path_analysis} that SRD, which is close to agony, is not statistically significantly asymmetric about zero \citep[see][for a non-linear measure of agony]{gupte2011finding}. 
We also do not observe a linear relationship between level differences in the organizational hierarchy and communication patterns: Figure~\ref{fig:eda_path}, whose vertical axis is on the log scale, shows that the relationship decreases nonlinearly away from zero.
Moreover, the relationship may not even be monotonically decreasing: there is more communication between high- and low-ranking employees than would be expected if communication decreased monotonically with reporting distance.
One interpretation of this finding is that there is less communication friction between nodes that are distant in the organizational tree than might be expected under an ``Agony'' theory of communication.  Instead, direct lines of communication emerge across long organizational distances, resulting in an emergent communication network with small-world properties.
This non-linearity and non-monotonicity also partially explains why the distance-based method from \citet{maiya2009inferring} may fail to accurately reconstruct the true organizational hierarchy.
Another possible reason is that the optimization algorithm in \citet{maiya2009inferring} is greedy and may arrive at a local solution that is not the global optimum.

Finally, the Frobenius distance, which counts edge and non-edge discrepancies equally necessarily takes small values for ``Distance" since a tree only has $n-1$ edges (out of $\binom{n}{2}$ possible edges) and the Distance reconstruction procedure tends to output a graph with fewer edges. On the other hand, Frobenius distance is large for Agony, since Agony produces a Eulerian graph that tends to have many edges.

\section{Discussion}\label{sec:discussion}

In this study, we find that in a large and successful software company, communication patterns have a strong association with organizational structure and that communication between employees increases with their organizational proximity, both globally across the organization (Figure~\ref{fig:email_density}) and locally within teams (Figure~\ref{fig:email_eda}).
Figure~\ref{fig:email_density} also provides evidence of the ``rich club" phenomenon at the top of the organization in which the most central employees are embedded in dense communities \citep{colizza2006detecting, dong2015inferring}.
Similarly, Figure~\ref{fig:eda_position} and Appendix Figure~\ref{fig:eda_centrality} show that measures of employees' centrality or importance in the communication network are larger for employees higher in the organizational chart, in agreement with several substantive theories of organizational behavior \citep{namata2006inferring, michalski2011matching, wang2013analyzing}.

Inspired by the use of hierarchical differences in the hierarchical position metric \citep{boeva2017analysis}, we have introduced new statistical measures of communication reciprocity.
We find that information flows asymmetrically -- more frequently up than down the organizational hierarchy -- and that reciprocation of communication depends on employees' position in the organization (Figure~\ref{fig:eda_position}).
In particular, high-ranking employees send emails that are reciprocated more often than they reciprocate emails.
This finding supports agony \citep{gupte2011finding} as a useful measure of organizational distance, since communication frequency decreases as the distance between employees increases.
However, we find that the emergent communication dynamics are much more complex than those predicted by existing theories. Consequently the reconstruction methods based on those overly simplistic theories inadequately capture the true relationships between these topological structures.
Agony, for example, supposes that communication in social networks is governed by a latent social hierarchy in which more popular individuals (higher in the social hierarchy) are less likely to communicate with individuals below them in the social hierarchy.
However, based on the lack of significance of our permutation test about the relationship between agony (equivalent to signed reporting distance) and communication frequency, as well as the poor performance of the agony-based reconstruction method in Figure~\ref{fig:tree}, this theory may not be an accurate model for how communication works in an organization, which is not surprising since we should expect information to flow both up and down in a functioning organization.
Similarly, \citet{maiya2009inferring} assume communication decreases monotonically with reporting distance, which is not what we find in Figure~\ref{fig:eda_path}.
Instead, a more realistic distance-based method for organizations might encourage ``small-world" links between upper management and low-level employees, as well as asymmetric communication up and down the hierarchy.

There are several limitations in this work.
We were unable to match 11\% of employees across the communication and organizational structures, which could affect team and global summary measures.
In particular, if the missing employees had roles or structural positions substantially unlike those of observed employees, this could result in selection bias that might alter the associations we found between network characteristics and role.
The temporal resolution of communication data is another possible limitation: although our measures of communication are observed rather than surveyed, the counts are aggregated over the entire month.
For this reason, we have not conducted a temporal/longitudinal analysis of communication dynamics.
We also do not have information about individual emails and multiple recipients, which prevents us from filtering out mass emails from super-senders \citep{guimera2006real, onnela2007structure} and from studying who initiates communication as a function of position in the hierarchy.
Similarly, we do not have meta-data on the employees and teams.
Most conspicuously, we are lacking productivity/performance metrics that would allow us to make concrete business suggestions about how to optimize the organizational structure if we had productivity measures on teams or employees.
However, this suggests a causal analysis in order to disentangle whether organization links facilitate communication, or if people in dense organizations are naturally more communicative.
Additionally, more contextual information (e.g. tenure at the company, previous positions within the company, projects assigned during the analysis time, etc.) would allow for more substantive interpretations of the results.
Finally, computations for the non-degree centrality measures, the reporting distances, and the reconstruction methods were not feasible on the entire email network, and it is possible that the individual team results are not reflective of the entire organization.

The study of organizational communication -- in large and successful businesses -- is still an emerging field. 
As we have shown, some theories of organization and communication capture important features of real-world communication patterns. But these theories may be too simplistic to adequately explain or predict organizational structures using communication data alone.  Empirical validation of organizational theories in a real data set, as we have tried to do here, is one way for management scholars to ensure that substantive hypotheses about communication dynamics match reality. Better measures of directional communication flow within organizational structures -- like the signed and directional reporting distances proposed above -- could help researchers understand these emergent communication dynamics.

Likewise, improved methods for reconstructing organizations should be rooted in realistic ideas about how directional communication dynamics relate to the organizational structure, rather than facile theories of employee popularity or importance.
Development of realistic explanatory models of communication conditional on an organizational tree might permit more accurate reconstruction of organizational trees.
The most promising direction is to generalize the distance-based approach in \citet{maiya2009inferring},
\begin{equation*}
    p(A_{uv} \mid T) = f(\alpha \cdot d_\beta(u, v)) \enskip ,
\end{equation*}
where $f$ is some function of some distance -- such as RD, SRD, or DRD -- between $u$ and $v$ in the tree $T$, and $\alpha$ and $\beta$ are hyperparameters.
This conditional likelihood suggests a Bayesian approach, which will require the development of priors on the space of trees.
The obvious bottleneck is the computational challenge: the pairwise distances need to be recomputed each time $T$ changes, which makes reconstruction at the scale of a large organization difficult.
Additionally, the evaluation and theoretical analysis of such methods will lag until there are more measures between trees.

By conditioning on a time-indexed organizational tree, a longitudinal analysis could reveal how changes in the organizational structure map to changes in communication, revealing individual responses to shock analogous to the team responses studied in \citet{athreya2022discovering}.
The addition of time will require care to track changes throughout the organization, both individuals that leave, join, or are promoted and teams that are added, removed, merged, or split.
This opens the door to the study of  retention, promotion, and turnover, and the role that position and distance have as a function of time.



Finally, our new measures, which provide a theory-based approach to measuring the
correspondence between an unknown or inferred hierarchy and a general network, are generalizable to any network analysis in which a latent hierarchy is known,
suspected, or desired.
Beyond the organizational literature, we believe this and future work may be useful in many other areas: social networks with latent social hierarchies; brain networks with latent
functional/structural hierarchies; biological/ecological correlation networks with latent phylogenetic trees; gene co-expression networks with latent gene hierarchies;
flow networks with latent structure.
These applications would require identifying the particulars of how the respective networks emerge as a function of their corresponding hierarchies, which we expect will differ from organizational hierarchies.
It is also possible that the full tree will not be available in other studies or applications.
Our measures will still apply to hierarchies at the community level, i.e. when nodes belong to communities that are hierarchically structured, but new methods may be needed for when only a partial tree or pairwise tree distances are available.

\section*{Acknowledgement}
We are grateful to the Editor and two reviewers for their valuable comments.
We are also grateful to Jonathan Larson and Carey Priebe for their helpful discussion of our work.
This work was supported by a grant from the National Institutes of Health Eunice Kennedy Shriver National Institute of Child Health and Development (1DP2HD091799-01).












\bibliographystyle{imsart-nameyear}
\bibliography{references.bib}

\pagebreak
\setcounter{equation}{0}
\setcounter{figure}{0}
\setcounter{table}{0}
\setcounter{page}{1}
\setcounter{section}{0}
\makeatletter
\renewcommand{\theequation}{S\arabic{equation}}
\renewcommand{\thefigure}{S\arabic{figure}}
\renewcommand{\bibnumfmt}[1]{[S#1]}
\renewcommand{\citenumfont}[1]{S#1}

\begin{frontmatter}
\title{Supplementary Material:\\ Communication network dynamics in a large organizational hierarchy}
\runtitle{Communication network dynamics}

\begin{aug}
\author[A]{\fnms{Nathaniel} \snm{Josephs}},
\author[B]{\fnms{Sida} \snm{Peng}},
\and
\author[C,D,E,F,G]{\fnms{Forrest W.} \snm{Crawford}}

\runauthor{Josephs et al.}

\address[A]{Department of Statistics, North Carolina State University}
\address[B]{Office of Chief Economist, Microsoft Research}
\address[C]{Department of Biostatistics, Yale School of Public Health}
\address[D]{Department of Statistics \& Data Science, Yale University}
\address[E]{Department of Ecology \& Evolutionary Biology, Yale University}
\address[F]{Yale School of Management}
\address[G]{RAND Corporation}
\end{aug}

\end{frontmatter}


\section{Analyses of communication and organizational structure in the Enron corpus}

The Enron email corpus consists of email communications among a group of senior employees of the firm in advance of its collapse in 2001 \citep{klimt2004enron}, and has been used to study the relationship between organizational hierarchy and communication patterns.
We summarize relevant analyses in Table~\ref{tab:enron}.

\begin{table}[H]
\resizebox{\textwidth}{!}{%
\begin{tabular}{@{}|l|l|l|@{}}
\toprule
\multicolumn{1}{|c|}{\textbf{Citation}}          & \multicolumn{1}{c|}{\textbf{Task}}    & \multicolumn{1}{c|}{\textbf{Method}}                    \\ \midrule
\rowcolor[HTML]{EFEFEF} 
\citet{shetty2004enron}       & Classify rank & Entropy                         \\
\citet{namata2006inferring}   & Classify rank & Email counts                    \\
\rowcolor[HTML]{EFEFEF} 
\citet{rowe2007automated}      & Classify rank & Flow and topological statistics \\
\citet{creamer2007segmentation} & Reconstruct organization              & Classify rank then add edges \\
\rowcolor[HTML]{EFEFEF} 
\citet{hossain2009effect}     & Classify rank & Centrality                      \\
\citet{zhang2009analyzing}      & Predict missing rank                  & Sent vs received email counts                           \\
\rowcolor[HTML]{EFEFEF} 
\citet{michalski2011matching} & Classify rank & Centrality                      \\
\citet{palus2011evaluation}     & Compare hierarchies & Hierarchical Position vs Social Score                   \\
\rowcolor[HTML]{EFEFEF} 
\citet{wang2013analyzing}     & Classify rank & PageRank modification           \\
\citet{dong2015inferring}     & Classify rank & Structural holes                \\
\rowcolor[HTML]{EFEFEF} 
\citet{nurek2020combining}    & Classify rank & Exogenous nodal information     \\ \bottomrule
\end{tabular}
}
\caption{Summary of the use of the Enron corpus \protect\citep{klimt2004enron} in the literature.}
\label{tab:enron}
\end{table}

\section{Scale-free analysis}\label{sec:scale_free}

A primary feature of network analyses is the study of the network's degree distribution, and one of the most prevalent beliefs is the ubiquity of scale-free networks \citep{albert2002statistical}.
However, a more recent survey concludes that scale-free networks are rare \citep{broido2019scale}.
Using the same methodology, we classify the the email networks at the team level, which is summarized in Figure~\ref{fig:scale_free}.
Note that in our analysis, we include every team even though only 9 of the 88 teams satisfies the original inclusion criteria from \citet{broido2019scale} that the mean degree is less than $\sqrt{n}$.

\begin{figure} 
    \centering
    \includegraphics[width=.5\textwidth]{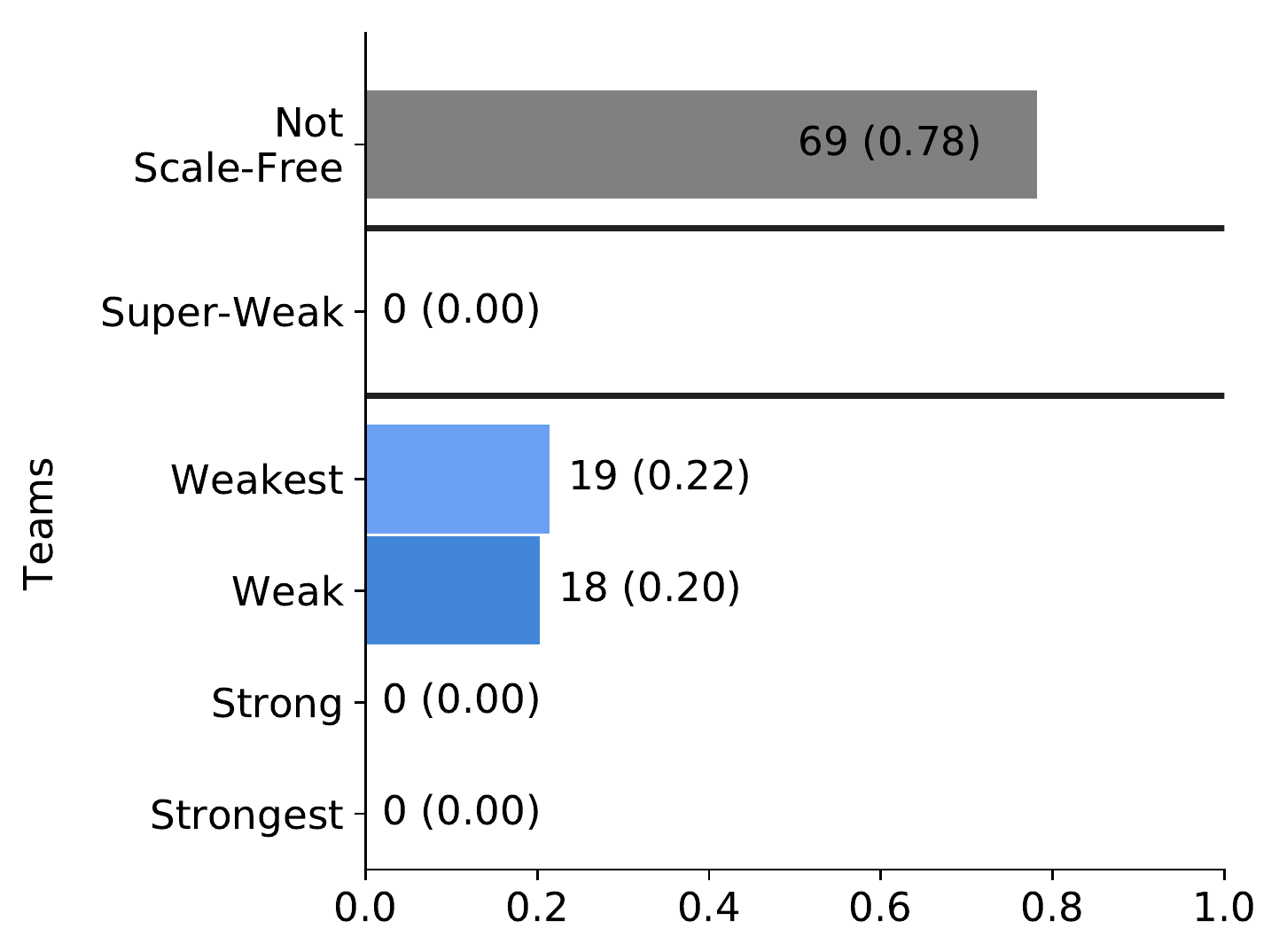}
    
    \caption{Scale-free analysis from \protect\citet{broido2019scale} performed on the team level.}
    \label{fig:scale_free}
\end{figure}

\section{Comparing network summary statistics}

There are many network summary statistics that we can compute for both the email network and organizational tree.
In Figure~\ref{fig:eda}, we show four of the clearest relationships between individual network summary statistics of a team's organizational tree (x-axis) and their corresponding email network (y-axis).
Starting at the top left and going clockwise, the Pearson correlations are $\rho~=-0.84, 0.75, 0.58$, and $-0.65$.

\begin{figure} 
    \centering
    \includegraphics[width = 0.49\textwidth]{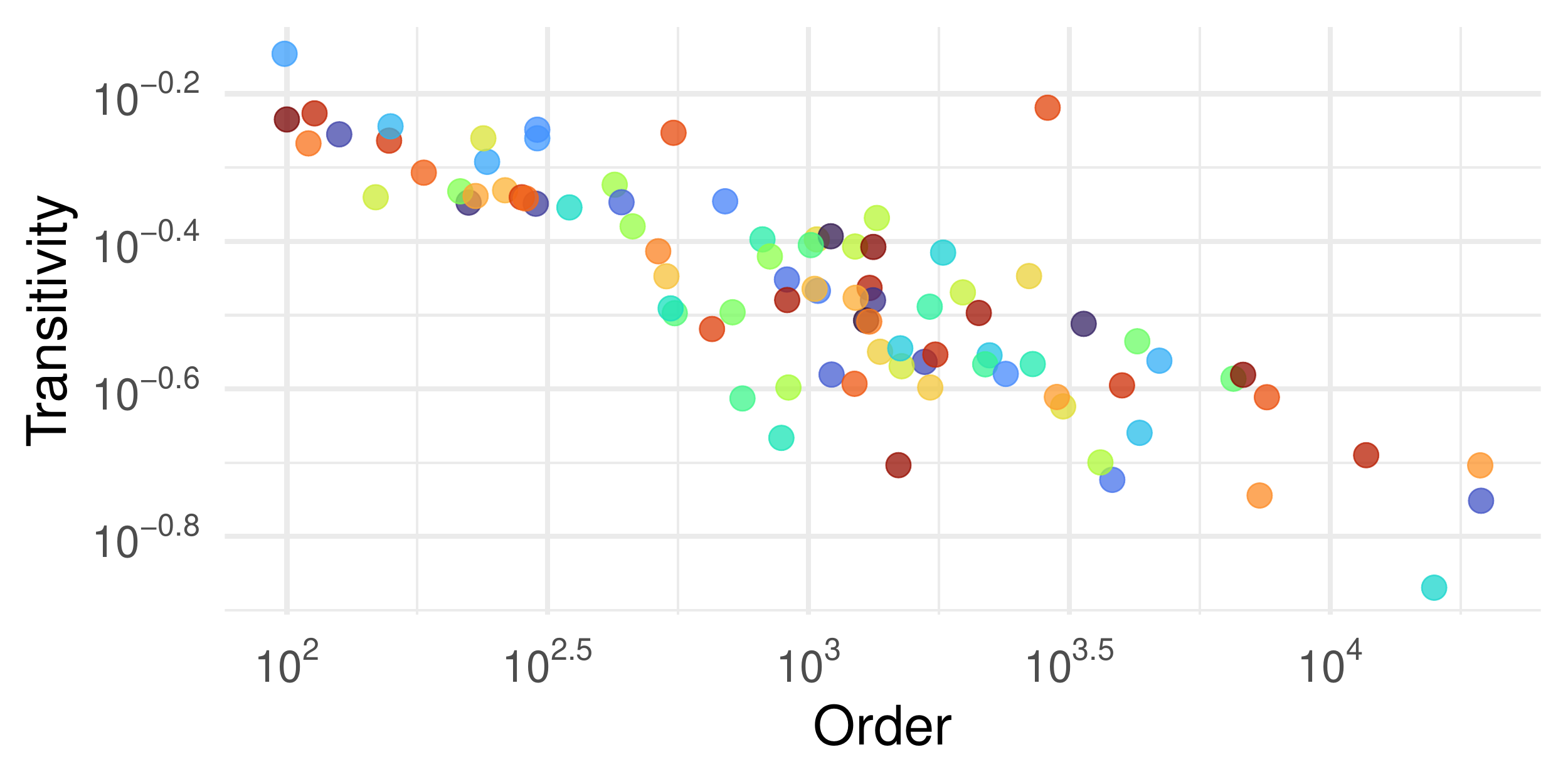}
	\includegraphics[width = 0.49\textwidth]{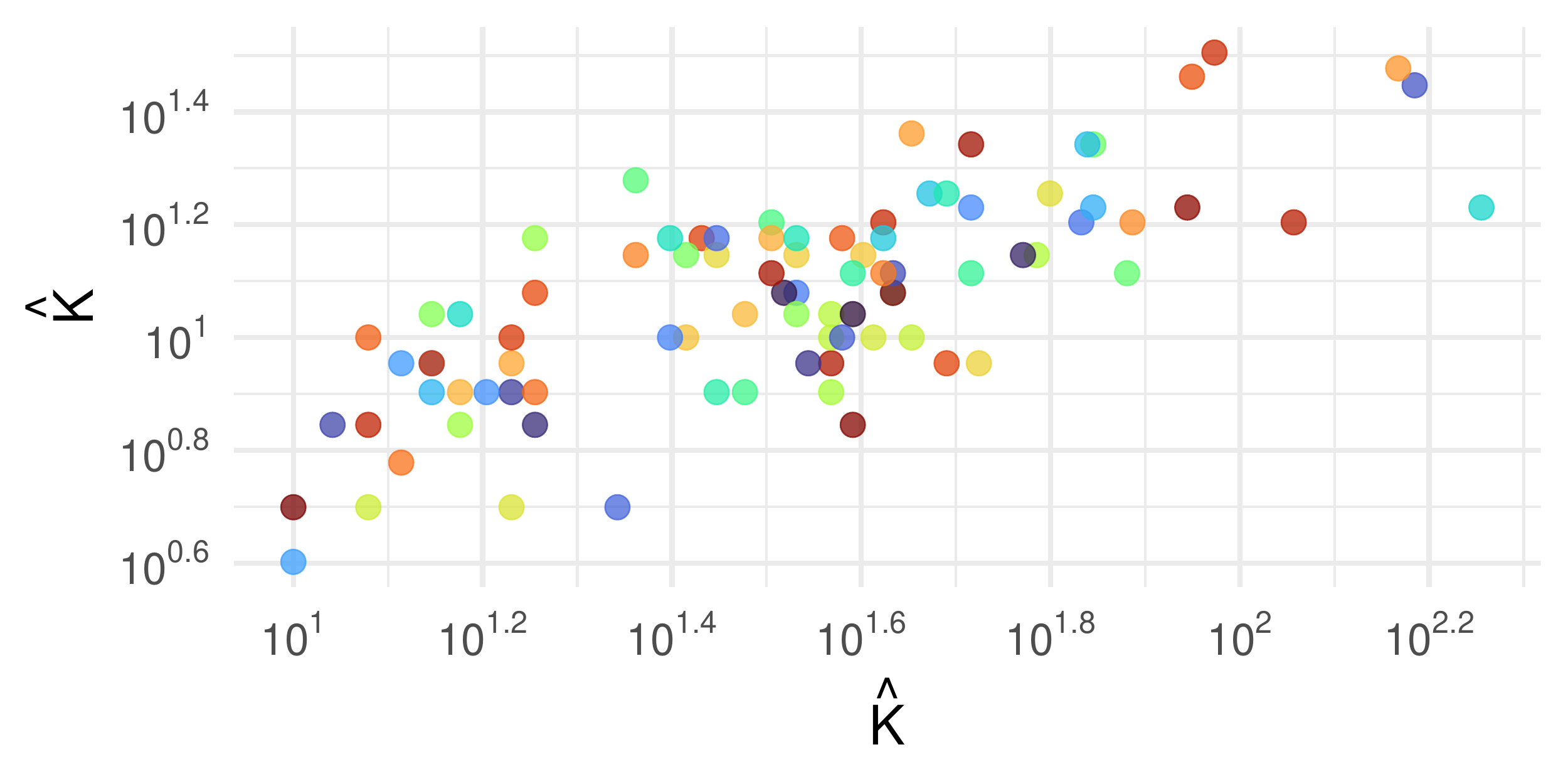}
	\includegraphics[width = 0.49\textwidth]{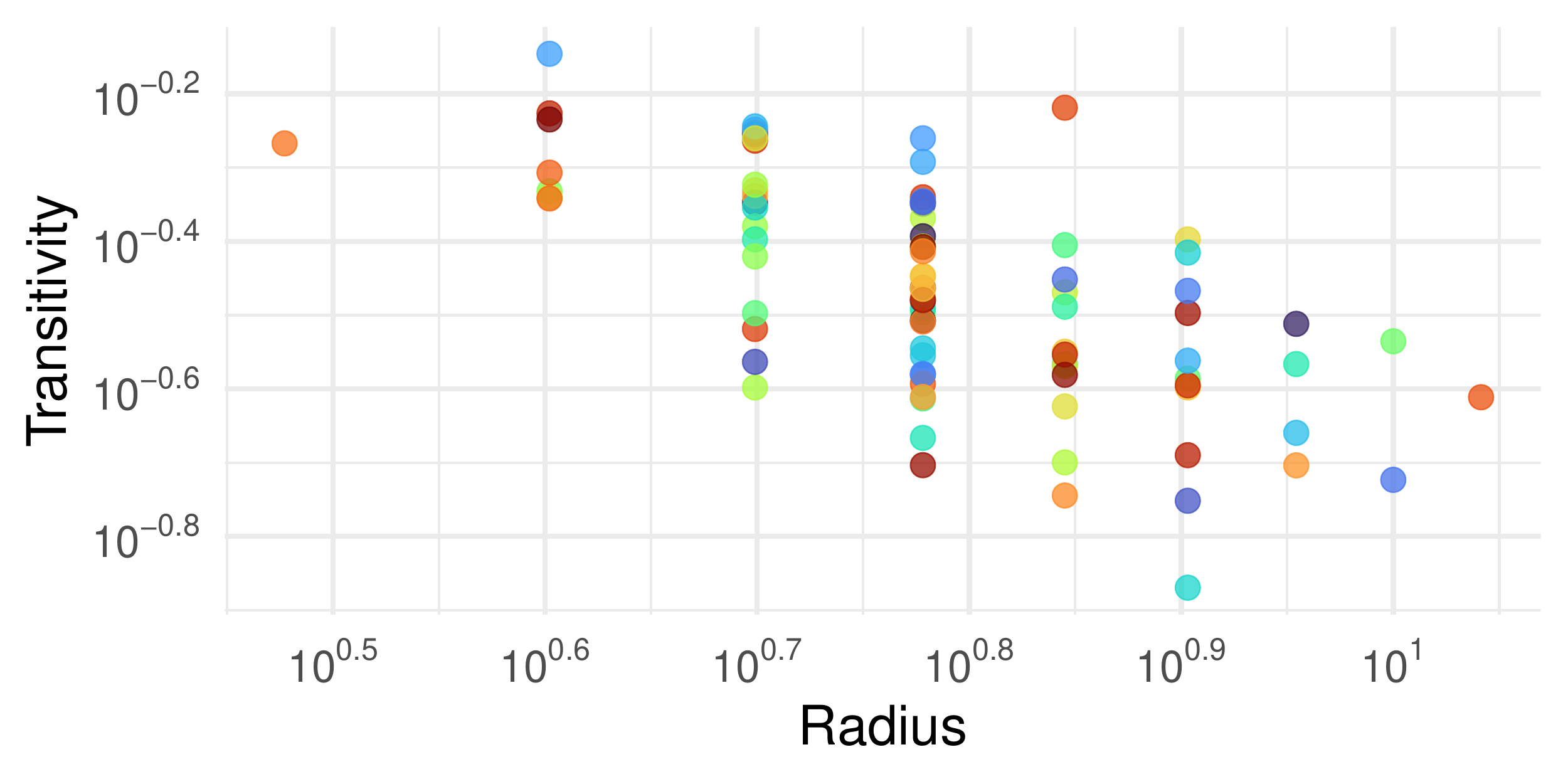}
	\includegraphics[width = 0.49\textwidth]{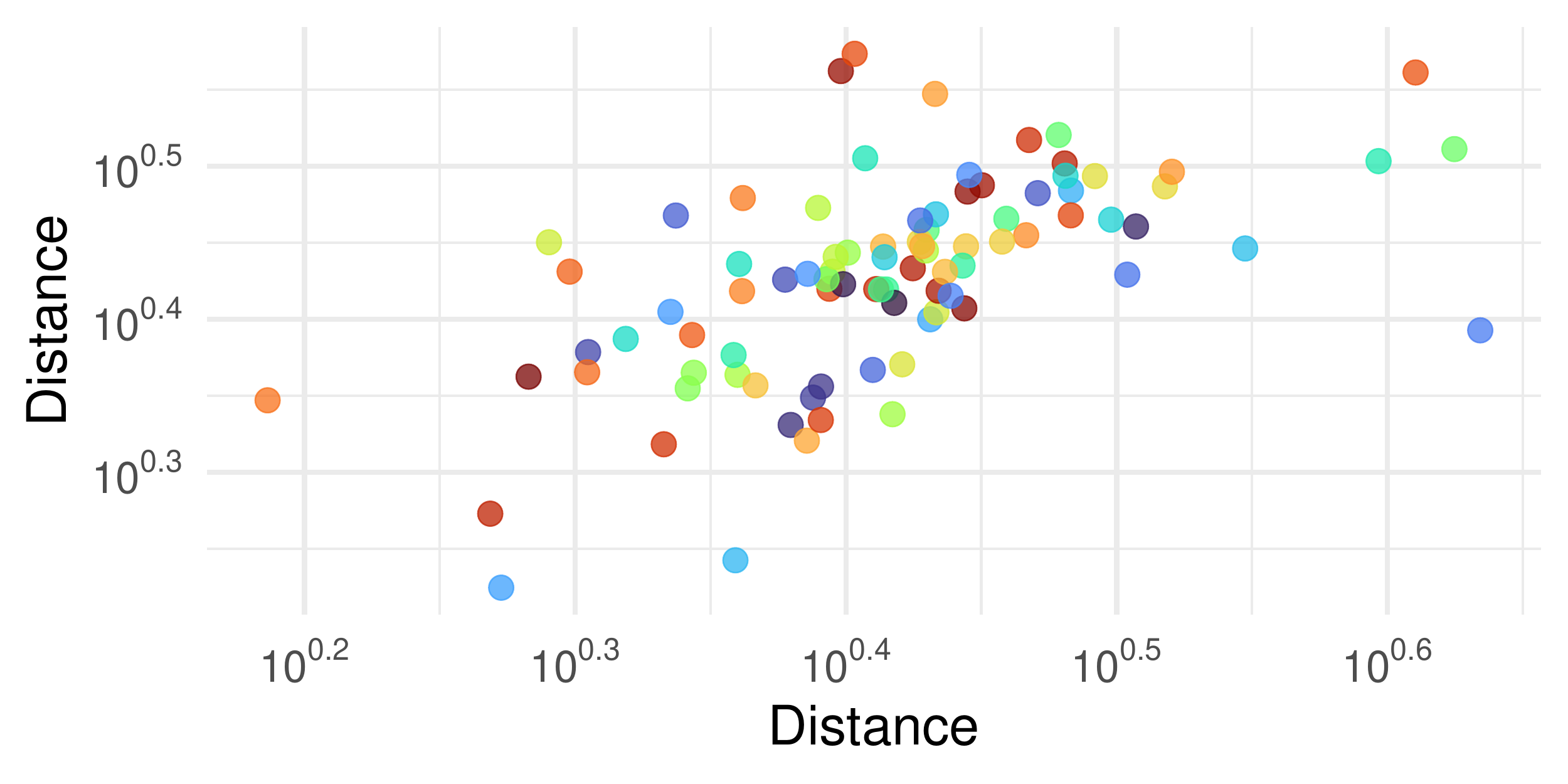}

    \caption{Various plots comparing the organizational tree (x-axis) and email network (y-axis). Points are colored by team as in Figure~\ref{fig:ms_vis}.}
    \label{fig:eda}
\end{figure}

\section{Importance in email network}

\subsection{SRD Hierarchical Position}

We can generalize \textit{hierarchical position} by replacing \textit{hierarchical difference} with SRD:
\begin{equation*}
    \text{SRD HP}(u) = \frac{\sum_{v \neq u} \text{SRD}(u,v)}{n_i - 1} \enskip .
\end{equation*}
This also allows us to define \textit{sent} and \textit{received} position in terms of SRD:
\begin{align*}
    \text{SRD SP}(u) &= \frac{\sum_{v \neq u} 1(A_{uv} > 0) \cdot \text{SRD}(u,v)}{\sum_v 1(A_{uv} > 0)} \\
    \text{SRD RP}(u) &= \frac{\sum_{v \neq u} 1(A_{vu} > 0) \cdot \text{SRD}(u,v)}{\sum_v 1(A_{vu} > 0)} \enskip .
\end{align*}
Analogous to Figure~\ref{fig:eda_position}, we present the difference of these proportions as a function of SRD HP in Figure~\ref{fig:hp_srd}.

\begin{figure} 
    \centering
    \includegraphics[width=.75\textwidth]{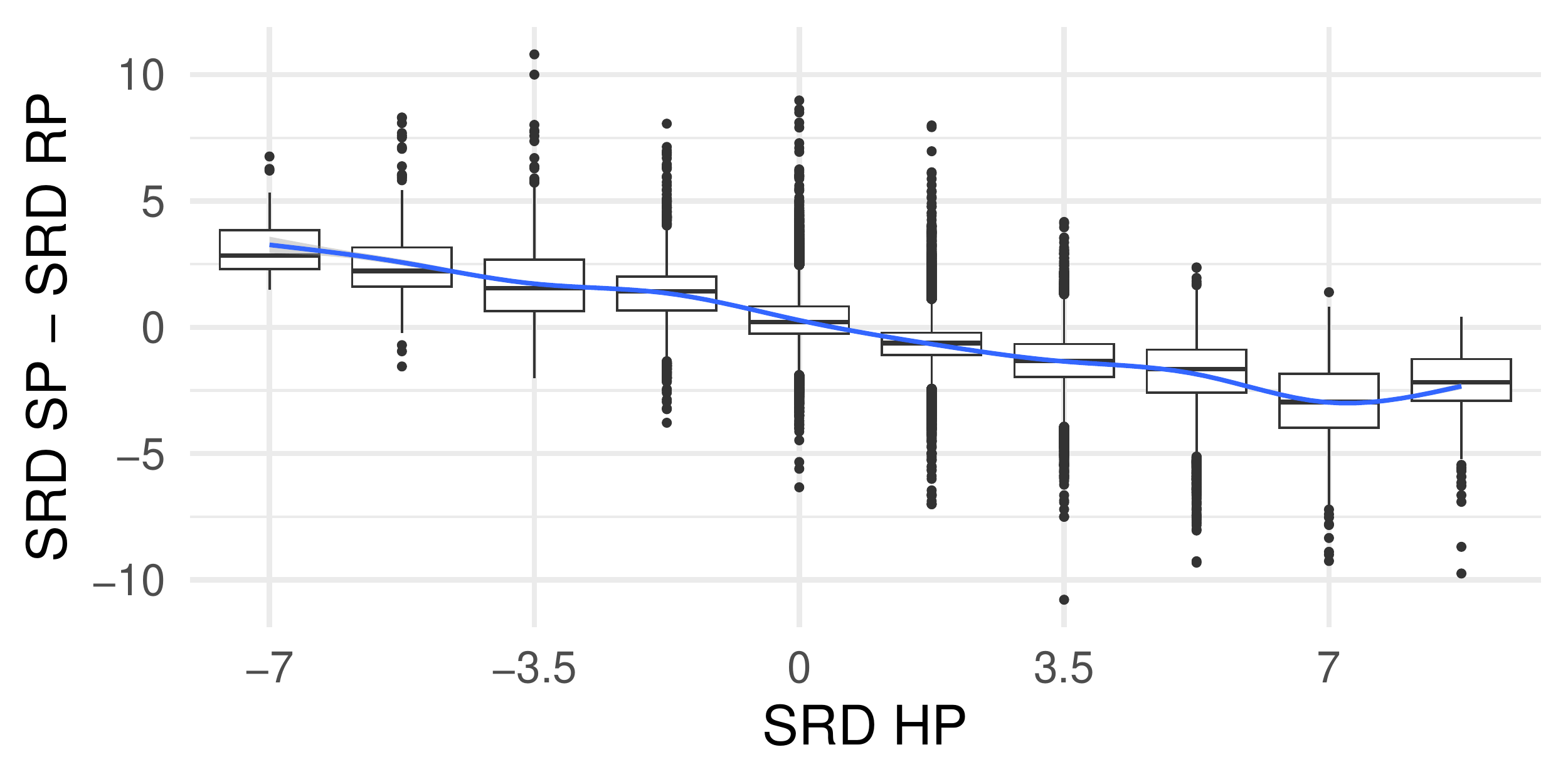}
    \caption{SP and RP by HP with SRD instead of hierarchical difference.. We bin position into 10 equal groupings and we show the box plot within each bin along with a smoothed curve (blue) of the individual (non-binned) data.}
    \label{fig:hp_srd}
\end{figure}

\subsection{Linear model of position on degree/strength}

Figure~\ref{fig:eda_position} shows that both degree and strength are negatively associated with position in the hierarchy.
That is, the lower an employee is in their team's hierarchy, the lower their degree/strength.
To assess whether this relationship is significant, we run four linear models of position against (log) degree and (log) strength.
Table~\ref{tab:importance} provides the coefficient estimates and standard errors, which reveals that all four relationships are significantly different than zero.

\begin{table}[H]
\resizebox{\textwidth}{!}{%
\begin{tabular}{@{}|l|r|r|@{}}
\toprule
\multicolumn{1}{|c|}{\textbf{Model}}          & \multicolumn{1}{c|}{$\widehat \beta_0$ (SE)}    & \multicolumn{1}{c|}{$\widehat \beta_1$ (SE)} \\ \midrule
\rowcolor[HTML]{EFEFEF} 
$\ \ \phantom{\log~}\text{Degree} = \beta_0 + \beta_1 \cdot \text{Level} / \text{Depth} + \varepsilon$ & 505.70 (4.66) & -374.22 (6.69) \\
$\ \ \log~\text{Degree} = \beta_0 + \beta_1 \cdot \text{Level} / \text{Depth} + \varepsilon$ & 1625.32 (14.73)  & -1117.26 (21.26) \\
\rowcolor[HTML]{EFEFEF} 
$\phantom{\log~}\text{Strength} = \beta_0 + \beta_1 \cdot \text{Level} / \text{Depth} + \varepsilon$ & 6.20 (0.02) & -1.90 (0.02) \\
$\log~\text{Strength} = \beta_0 + \beta_1 \cdot \text{Level} / \text{Depth} + \varepsilon$ & 7.35 (0.02) & -1.89 (0.03) \\ \bottomrule
\end{tabular}
}
\caption{Summary of linear models assessing the relationship between position and degree/strength in Figure~\ref{fig:eda_position}.}
\label{tab:importance}
\end{table}

\subsection{Centrality measures}

We want to assess whether importance in the email network corresponds to importance in the organization.
We assess this for each team and importance in the team's organization is defined as relative position.
For importance in the email network, we consider four measures of centrality -- betweenness, closeness, eigenvalue, and authority \citep[see][for definitions]{kolaczyk2014statistical} -- which is shown in Figure~\ref{fig:eda_centrality}.

\begin{figure} 
    \centering
    \includegraphics[width=\textwidth]{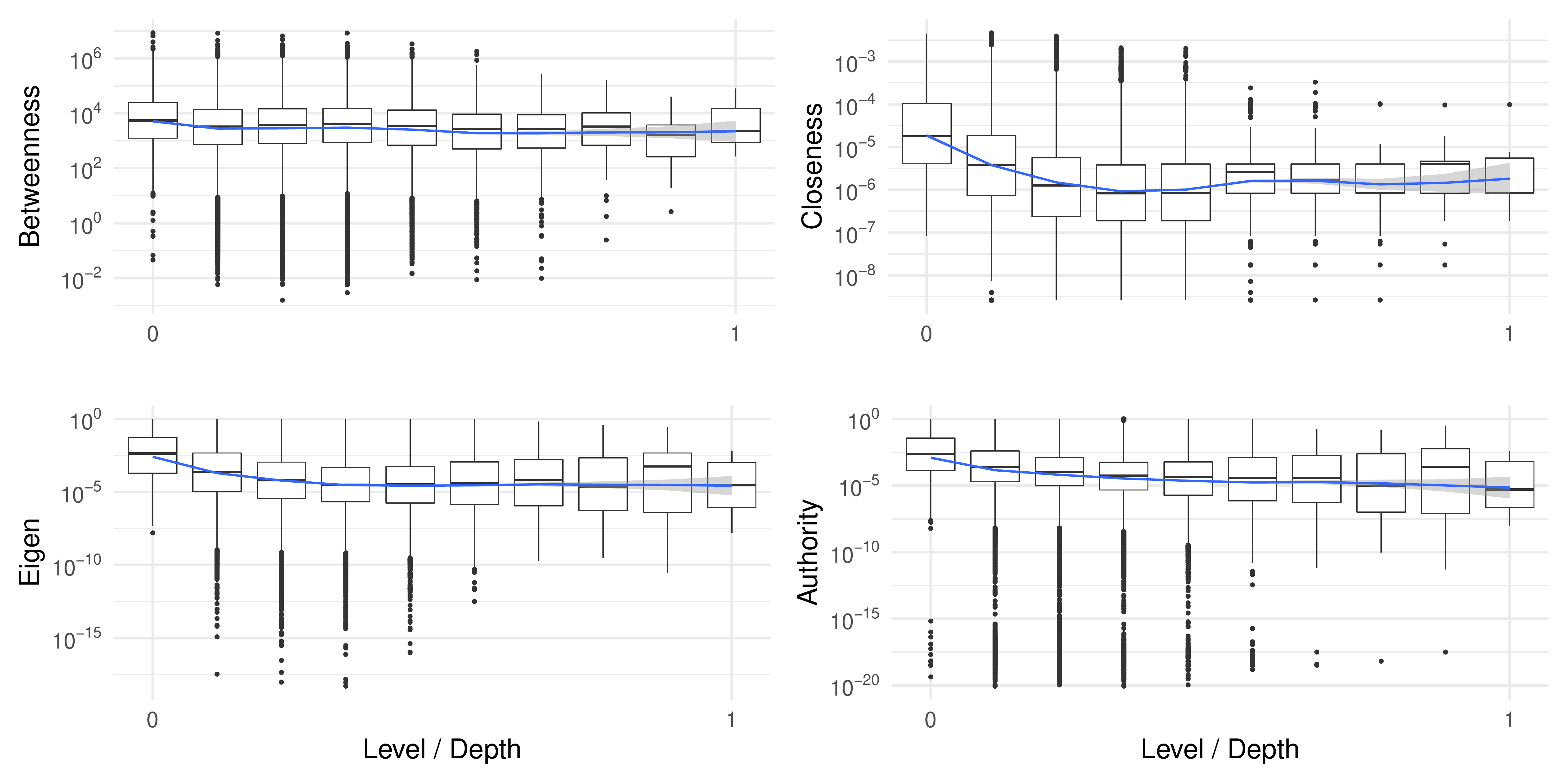}
    
    \caption{Centrality measures by relative position in the hierarchy. In all of the plots, the relative positions have been binned into 10 equal groupings. We show the box plot within each bin along with a smoothed curve (blue) of the raw data.}
    \label{fig:eda_centrality}
\end{figure}

\section{Path analysis}

\subsection{Computing $n^\text{up}$ and $n^\text{down}$}\label{sec:n_up_down}

In a tree $T$, there is a unique shortest path between any two nodes $u, v \in V(T)$.
This shortest path can take steps both up and down the tree, and hence its length is the sum of the number of steps up and down, $n^\text{up}(u, v) + n^\text{down}(u, v)$.
Dijkstra's algorithm is used to compute the shortest path lengths between the nodes in a graph, and can be modified to compute $n^\text{up}(u, v)$ and $n^\text{down}(u, v)$ as follows.

First, assume $T$ is directed down, i.e. the in-degree of the root node is 0 and all other nodes have in-degree equal to 1.
Create the transpose of $T$ that is directed up, say $T^T$, and then create the graph $G = T \cup T^T$.
Note that the adjacency matrix of $G$ is now symmetric.
We equip weights to the edges of $G$ such that,
\begin{equation*}
    G_{uv} = \begin{cases}
        p & \text{if }~ T_{uv} = 1 \\
        q & \text{if }~ T^T_{uv} = 1 \\
        0 & \text{else}
    \end{cases} \enskip ,
\end{equation*}
where $p$ and $q$ are prime and $p >> q$.
Next, run Dijkstra's algorithm on $G$, which will return the shortest (weighted) path lengths between all nodes, say $d(u, v)$.
We can finally recover the desired quantities as
\begin{align*}
    n^\text{up}(u, v) &= \Big(d_{uv} \pmod{p}\Big) / q \\
    n^\text{down}(u, v) &= \Big(d_{uv} - q \cdot n^\text{up}(u, v)\Big) / p \enskip .
\end{align*}
In our calculations, we take $p = 101$ and $q = 3$.

\subsection{Reporting distances}\label{sec:rep_dist}

Both signed reporting distance and directed reported distance can be negative, so neither is a true metric.
Moreover, neither satisfies the identity of indiscernibles, since any two nodes that are on the same level will have SRD and DRD equal to zero.
Likewise, neither satisfies symmetry.
However, they both satisfy a negative symmetry since $\text{SRD}(u,v) = -\text{SRD}(v, u)$ and $\text{DRD}(u,v) = -\text{DRD}(v, u)$.

DRD does not satisfy the triangle inequality.
As one counterexample, suppose $w$ reports directly to $v$, and both $u$ and $v$ are direct reports to some fourth node.
Then $\text{DRD}(u, v) = 0$, but $\text{DRD}(u, w) = -3$ and $\text{DRD}(w, v) = 1$.

Interestingly, SRD satisfies the triangle equality.

\begin{proposition}\label{prop:tri_eq}
    Let $T$ be a tree. For any $u, v, w \in V(T)$,
    \begin{equation*}
        \text{SRD}(u, v) = \text{SRD}(u, w) + \text{SRD}(w, v) \enskip .
    \end{equation*}
\end{proposition}

\begin{proof}

Let $u, v, w \in V(T)$.
If any two nodes are the same, then the result is trivial since $n^\text{up}(u, u) = n^\text{down}(u, u) = 0$.

Next, let $u, v, w \in V(T)$ all be different.
Suppose $w$ lies on the shortest path from $u$ to $v$.
It follows that $\text{RD}(u, v) = \text{RD}(u, w) + \text{RD}(w, v)$.
Thus
\begin{alignat*}{2}
    &&n^\text{up}(u, v) + n^\text{down}(u, v) &= n^\text{up}(u, w) + n^\text{down}(u, w) + n^\text{up}(w, v) + n^\text{down}(w, v) \\
    &\implies \quad &n^\text{up}(u, v) - n^\text{down}(u, v) &= n^\text{up}(u, w) + n^\text{down}(u, w) + n^\text{up}(w, v) + n^\text{down}(w, v) - 2n^\text{down}(u, v) \\
    &\implies \quad &\text{SRD}(u, v) &= n^\text{up}(u, w) - n^\text{down}(u, w) + n^\text{up}(w, v) - n^\text{down}(w, v) \\
    &&&\ \ \ - 2\Big(n^\text{down}(u, v) - n^\text{down}(u, w) - n^\text{down}(v, w)\Big) \\
    &&&= \text{SRD}(u, w) + \text{SRD}(w, v) \enskip .
\end{alignat*}
The last equality is due to the fact that, since $w$ lies on the shortest path from $u$ to $v$, $n^\text{down}(u, v) = n^\text{down}(u, w) + n^\text{down}(v, w)$.

On the other hand, suppose $w$ does not lie on the shortest path from $u$ to $v$.
We proceed by induction on the length of the shortest path from $u$ to $w$.
If $\text{RD}(u, w) = 1$, then, $u$ and $w$ are connected.
Without loss of generality, say $n^\text{up}(u, w) = 1$ and $n^\text{down}(u, w) = 0$.
Then the shortest path from $w$ to $v$ is just one edge longer than the shortest path from $u$ to $v$.
That is, $\text{RD}(w, v) = 1 + \text{RD}(u, v)$.
Thus
\begin{alignat*}{2}
    &&n^\text{up}(u, v) + n^\text{down}(u, v) &= n^\text{up}(w, v) + n^\text{down}(w, v) - 1 \\
    &\implies \quad &n^\text{up}(u, v) - n^\text{down}(u, v) &= n^\text{up}(w, v) + n^\text{down}(w, v) - 1 - 2n^\text{down}(u, v)\\
    &\implies \quad &\text{SRD}(u, v) &= n^\text{up}(u, w) - n^\text{down}(u, w) + n^\text{up}(w, v) - n^\text{down}(w, v) \\
    &&&\ \ \ - 2 - 2n^\text{down}(u, v) + 2 n^\text{down}(w, v) \\
    &&&= \text{SRD}(u, w) + \text{SRD}(w, v) \enskip .
\end{alignat*}
The penultimate equality uses the fact that $n^\text{up}(u, w) - n^\text{down}(u, w) = 1$
To see the last equality, observe that since $n^\text{up}(u, w) = 1$, we know that $u$ must report directly to $w$.
Furthermore, because $w$ does not lie on the path from $u$ to $v$, the first step on the shortest path from $u$ to $v$ must go down.
Therefore, the shortest path from $w$ to $v$ must go through $u$.
It follows that $n^\text{down}(w, v) = n^\text{down}(u, v) + 1$, hence
\begin{equation*}
    2 - 2n^\text{down}(u, v) + 2 n^\text{down}(w, v) = 2\cdot\big(1 - n^\text{down}(u, v) + n^\text{down}(w, v)\big) = 0 \enskip .
\end{equation*}
With our basis for induction established, assume that $\text{SRD}(u, v) = \text{SRD}(u, w) + \text{SRD}(w, v)$ holds for $\text{RD}(u, w) = k$.
We are done if we can show the result holds when $\text{RD}(u, w) = k + 1$.

If $\text{RD}(u, w) = k + 1$, then there exists a point $z$ on the shortest path from $u$ to $w$.
Since $\text{RD}(u, w) = \text{RD}(u, z) + \text{RD}(z, w)$, we know that $\text{RD}(u, z) \leq k$ and $\text{RD}(z, w) \leq k$.
By assumption, it follows that $\text{SRD}(u, v) = \text{SRD}(u, z) + \text{SRD}(z, v)$ and $\text{SRD}(z, v) = \text{SRD}(z, w) + \text{SRD}(w, v)$.
Putting these together, we have
\begin{align*}
    \text{SRD}(u, v) &= \text{SRD}(u, z) + \text{SRD}(z, w) + \text{SRD}(w, v) \\
    &= n^\text{up}(u, z) - n^\text{down}(u, z) + n^\text{up}(z, w) - n^\text{down}(z, w) + \text{SRD}(w, v) \\
    &= n^\text{up}(u, w) - n^\text{down}(u, w) + \text{SRD}(w, v) \\
    &= \text{SRD}(u, w) + \text{SRD}(w, v) \enskip .
\end{align*}
Note that the penultimate equality follows from the fact that $z$ is on the shortest path from $u$ to $w$.
Our induction is complete and concludes our proof.
\end{proof}

\subsection{Communication frequency}

Figure~\ref{fig:eda_path_unweighted} shows the communication frequency between pairs of employees as a function of the reporting distances.
As in Figure~\ref{fig:eda_path}, on average, the closer two employees are in the organizational tree, the more frequently they communicate.

\begin{figure} 
    \centering
    \includegraphics[width=\textwidth]{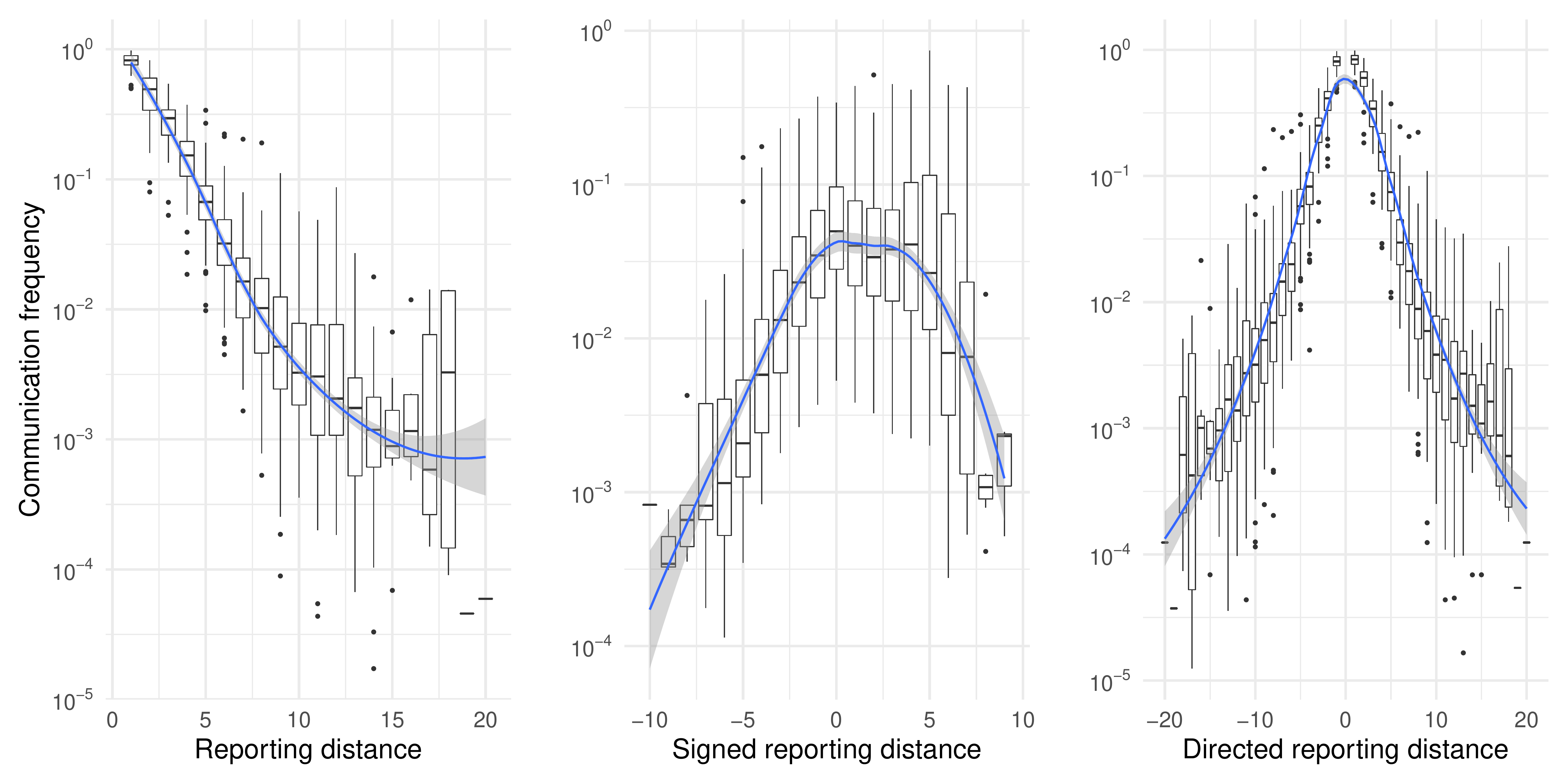}
    \caption{Pairwise reporting distances in the organizational tree and the communication frequency among all pairs in that reporting distance. Reporting distances are computed within each team and the box plots summarize the results across all of the teams.}
    \label{fig:eda_path_unweighted}
\end{figure}

\subsection{Team path analysis}

In Figure~\ref{fig:eda_path}, we summarized the average number of emails exchanged across teams in boxplots.
In Figure~\ref{fig:eda_path_team}, we show the same plot broken out by individual team.

\begin{figure} 
    \centering
    \includegraphics[width=\textwidth]{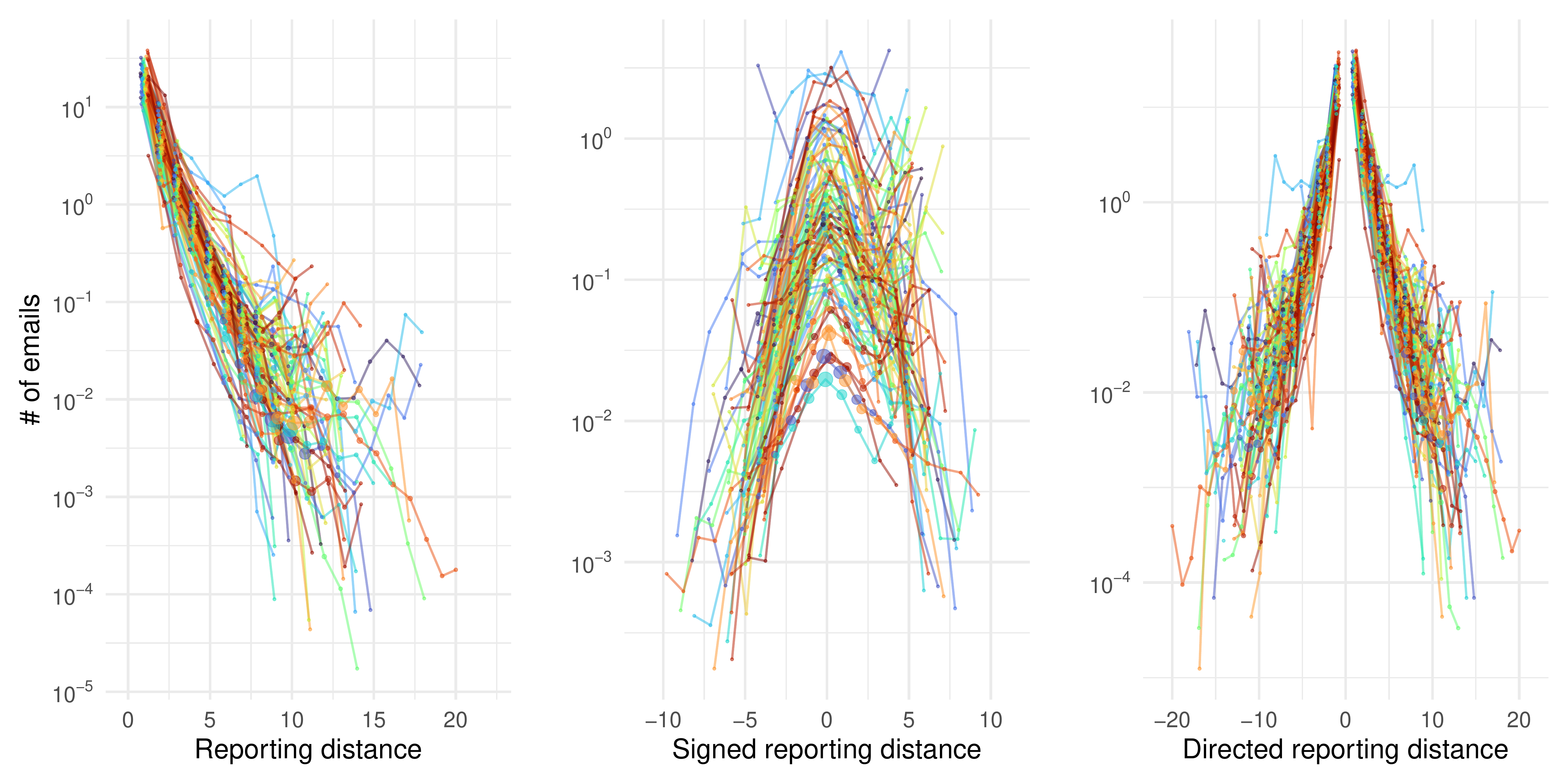}
    \caption{Plots showing pairwise reporting distances in a team's organizational tree and the average number of emails exchanged among all pairs of that reporting distance. Lines are colored by team as in Figure~\ref{fig:ms_vis}.}
    \label{fig:eda_path_team}
\end{figure}

\subsection{Permutation tests}

We performed two permutation tests to see whether the right and center plots of Figure~\ref{fig:eda_path} are symmetric about zero.
Recall that our test statistic is
\begin{equation*}
    t(A) = \sum_{k=1}^{k_{\max}} \Big(\frac{1}{\vert S_k \vert}\sum_{(u, v) \in S_k}A_{uv} - \frac{1}{\vert S_{-k} \vert}\sum_{(u, v) \in S_{-k}}A_{uv}\Big)^2  \enskip ,
\end{equation*}
where $S_k = \{(u,v) : \text{DRD}(u, v) = k\}$ and $k_{\max} = \max_{u,v} \text{DRD}(u, v)$.
In Figure~\ref{fig:eda_path_perm}, we show our null distributions, which we obtained by randomly permuting the number of emails exchanged among pairs whose \textit{reporting distances} are the same.
The null distributions are aggregated over $500$ different permutations within each team, except for the 4 teams that had $> 10,000$ members, whose permutations were too computationally expensive.
The empirical test statistics are given by the red lines.

\begin{figure} 
    \centering
    \includegraphics[width=.49\textwidth]{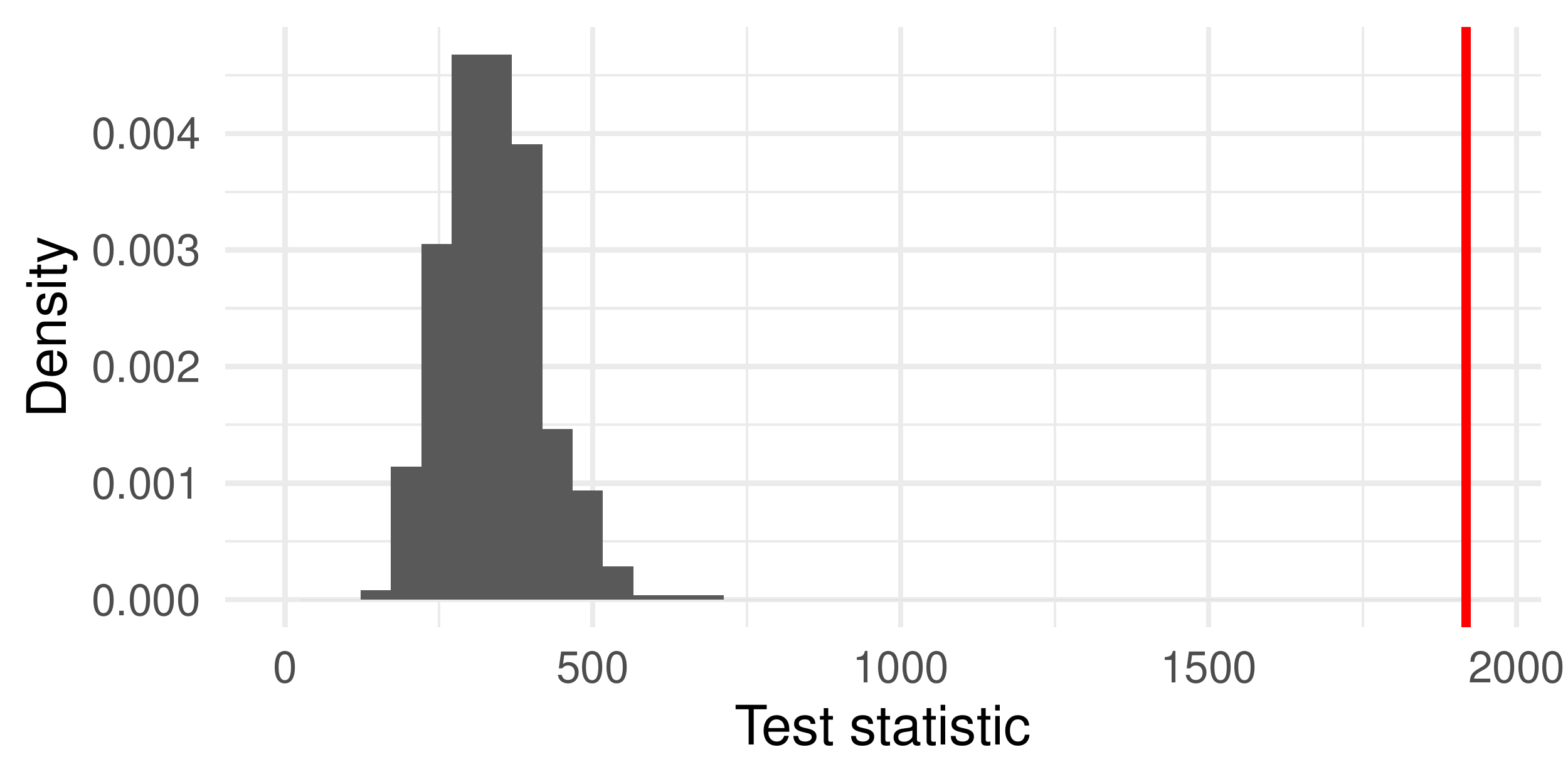}
    \includegraphics[width=.49\textwidth]{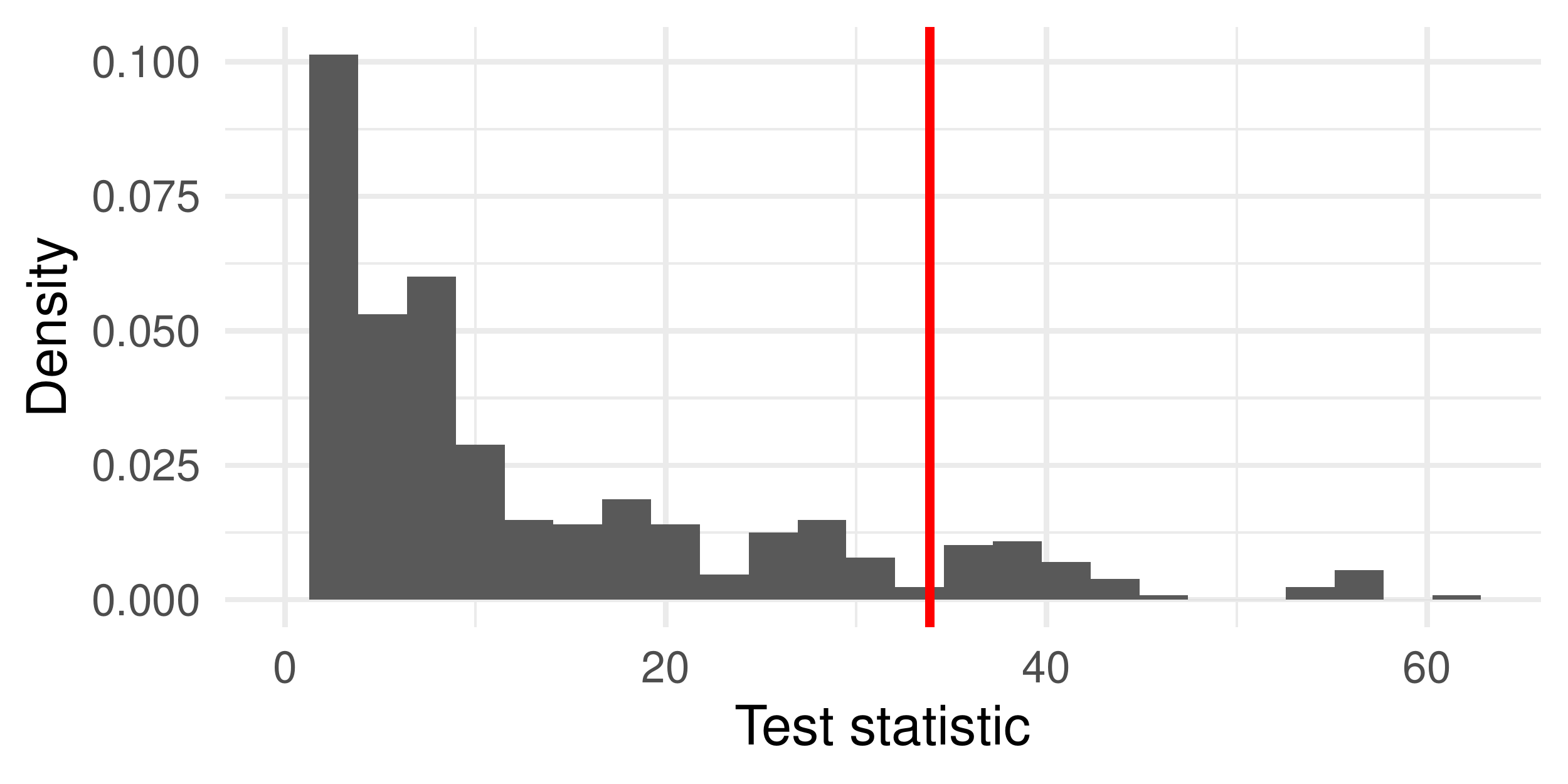}
    \caption{Null distributions from the permutation tests of symmetry described in Section~\ref{sec:path_analysis}. (Left) Directed reporting distance. (Right) Signed reporting distance. In both plots, the red line is the observed test statistic.}
    \label{fig:eda_path_perm}
\end{figure}

\section{Tree distances}\label{sec:tree_dist}

Let $T$ be a directed rooted tree as in the organizational tree of the main text.
Thus $T_{uv} = 1$ if $u$ reports to $v$ and 0 otherwise.
Let $\hat{T}$ be an estimate of $T$.
Finally, let $c_u$ and $\hat{c}_u$ be the betweenness centrality for node $u$ in $T$ and $\hat{T}$, respectively.

We utilize the following two distances to measure how well $\hat{T}$ estimates $T$:
\begin{align}
    d_\text{F}\big(T, \hat{T}\big) &= \frac{1}{n-1} \sqrt{\sum_{u=1}^n \sum_{v=1}^n (T_{uv} - \hat{T}_{uv})^2} \tag{Frobenius} \\
    d_\text{cent}\big(T, \hat{T}\big) &= \frac{1}{n-1} \sqrt{\sum_{u=1}^n (c_u - \hat{c}_v)^2} \tag{\citet[Eq 4.2]{donnat2018tracking}}
\end{align}
The first measures how many edges are in agreement between the estimated tree, $\hat{T}$, and the true tree, $T$.
The second measures the agreement of centrality measures for the nodes in each graph and ``a change in centrality can be understood as a drift of the node away from (or toward) the core of the network" \citep{donnat2018tracking}.

Note that both of these distances can also be generally applied to measure the similarity of two arbitrary graphs.

\section{Tree reconstruction}\label{sec:tree_recon}

The absolute performance of the reconstruction methods in Section~\ref{sec:model} is given in Figure~\ref{fig:tree_abs}.

\begin{figure} 
    \centering
    \includegraphics[width=.75\textwidth]{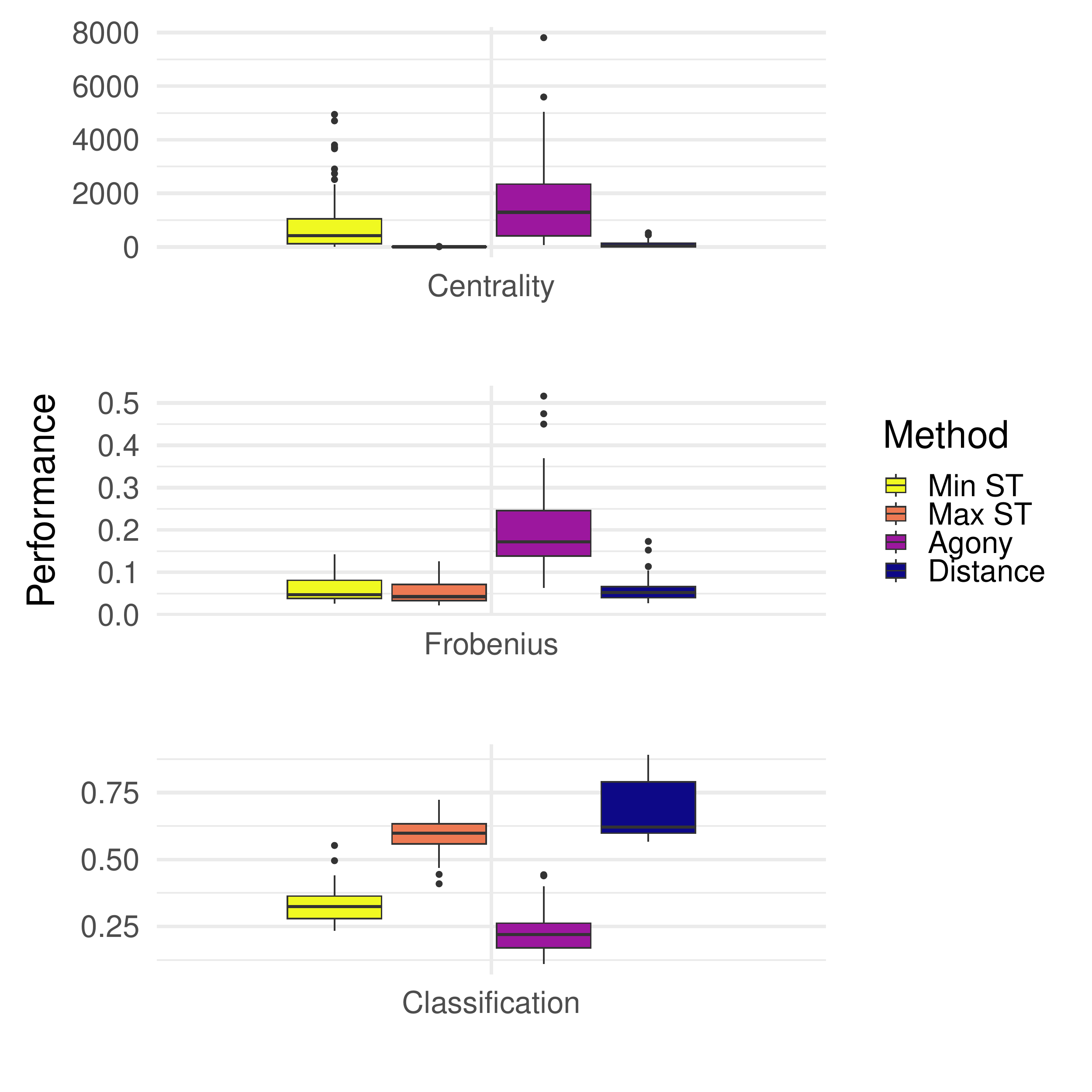}
    \caption{Absolute performance of methods for reconstructing the organizational trees from the communication networks: Min and Max ST \citep{prim1957shortest}, Distance \citep{maiya2009inferring}, and Agony \citep{gupte2011finding}.}
    \label{fig:tree_abs}
\end{figure}

Figure~\ref{fig:tree_levels} shows the mean squared error (MSE) of the predicted level versus the true level as a function of the true level.
That is, we plot level $k$ against
\begin{equation*}
    \frac{1}{\vert u: z(u) \leq k\vert} \sum_{u: z(u) \leq k} \big(z(u) - \hat{z}(u)\big)^2
\end{equation*}
where $z(u)$ is the level of node $u$.
We note that while Min ST and Max ST produce a tree, it is not rooted.
Therefore, we define $\hat{z}(u)$ as the distance in the estimated tree from the true root.
We also do not include Agony and Distance because neither method produces a tree.
We see that Min ST and Max ST are both better at recovering the top of the hierarchy than the lower levels.

\begin{figure} 
    \centering
    \includegraphics[width=.75\textwidth]{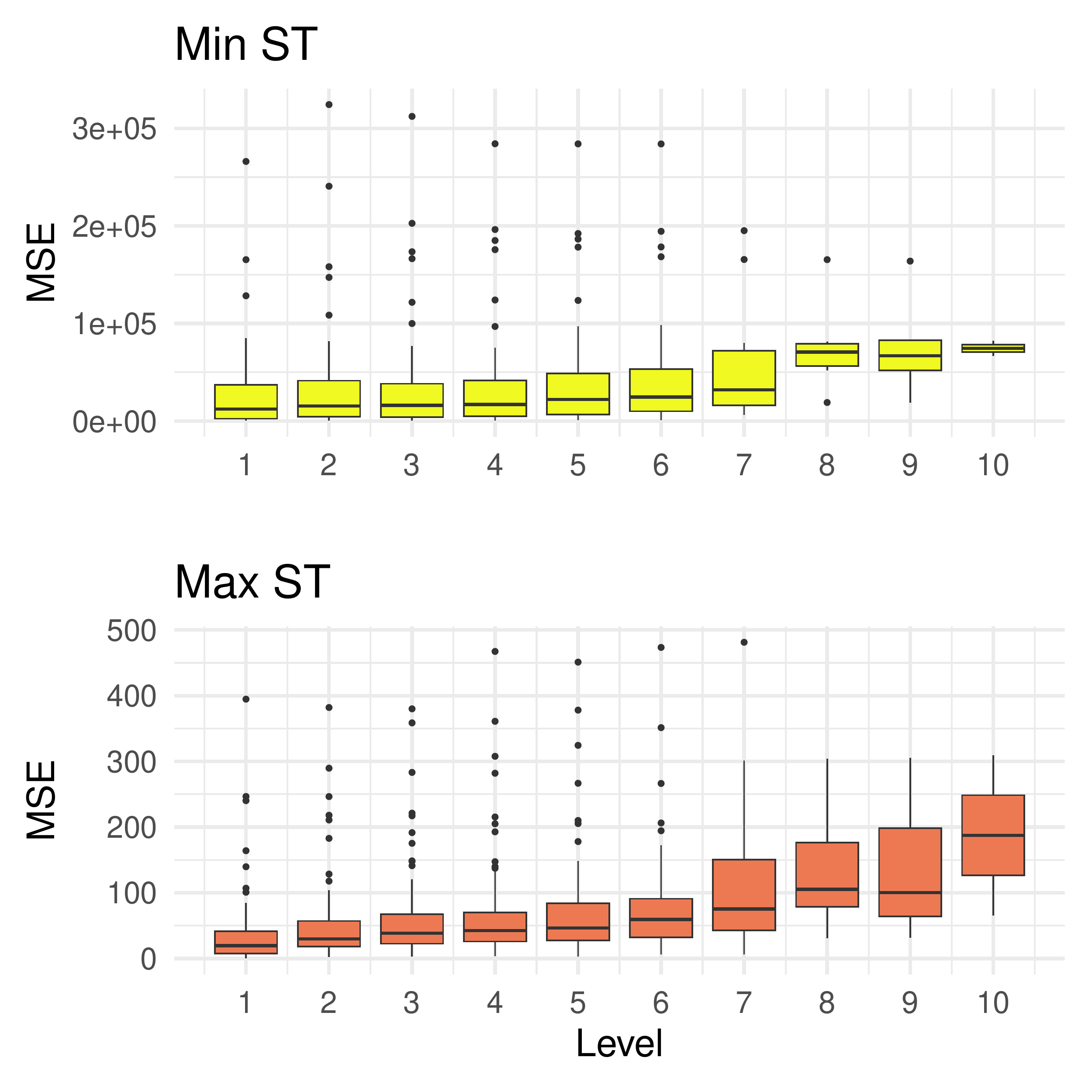}
    \caption{Mean squared error (MSE) of the predicted level versus the true level as a function of level.}
    \label{fig:tree_levels}
\end{figure}


\end{document}